\documentclass[10pt,a4paper,reqno]{amsart}
\usepackage{amsmath,amsthm,amsfonts,amssymb,euscript,bbm,color,slashed}

\usepackage{marvosym,accents,slashed}
\usepackage{graphicx}

\usepackage{xcolor}

\newtheorem{theorem}{Theorem}[section]
\newtheorem{proposition}[theorem]{Proposition}
\newtheorem{lemma}[theorem]{Lemma}

\newtheorem{definition}[theorem]{Definition}

\newtheorem{problem}[theorem]{Problem}

\newtheorem*{remark*}{Remark}

\usepackage[margin=0.6in,dvips]{geometry}

\numberwithin{equation}{section}
\numberwithin{theorem}{section}

\usepackage{bbm}
\usepackage{graphics}
\usepackage{enumerate}

\def\th {\vartheta}

\def\ep {\epsilon}

\def\f {\frac}

\def\rd {\partial}
\def\ls {\lesssim}
\def\de {\delta}

\def\Omg{\Omega}

\newcommand{\ud}{\mathrm{d}}

\def\bnab{\overline{\partial}}
\def\srd {\slashed{\partial}}
\def\sg{\slashed{g}}

\newcommand{\pfstep}[1]{\vspace{.5em} {\it \noindent #1.}}

\newcommand{\bea}{\begin{eqnarray}}
\newcommand{\eea}{\end{eqnarray}}
\def\beaa{\begin{eqnarray*}}
\def\eeaa{\end{eqnarray*}}

\newcommand{\tr}{\mathrm{tr}}
\renewcommand{\div}{\mathrm{div}}

\begin{document}

\title[Non-existence of trapped surfaces]{On the non-existence of trapped surfaces \\
under low-regularity bounds}
\author{Jonathan Luk}
\address{Department of Mathematics, Stanford University, 450~Serra~Mall~Building~380,~Stanford~CA~94305-2125,~USA}
\email{jluk@stanford.edu}
\author{Georgios Moschidis}
\address{Department of Mathematics, Princeton University, 509 Fine Hall, Washington Rd, Princeton~NJ~08544,~USA}
\email{gm6@math.princeton.edu}

\maketitle
\vspace{-2ex}
\begin{center}{\it\large Dedicated to Professor Demetrios Christodoulou, with admiration}
\end{center}

\begin{abstract}
    The emergence of trapped surfaces in solutions to the  Einstein field equations is intimately tied to the well-posedness properties of the corresponding Cauchy problem in the low regularity regime.  In this paper, we study the question of existence of trapped surfaces already at the level of the initial hypersurface when the scale invariant size of the Cauchy data is assumed to be bounded. Our main theorem states that no trapped surfaces can exist initially when the Cauchy data are close to the data induced on a spacelike hypersurface of Minkowski spacetime (not necessarily a flat hyperplane) in the Besov $B^{3/2}_{2,1}$ norm. We also discuss the question of extending the above result to the case when merely smallness in $H^{3/2}$ is assumed.
\end{abstract}

\section{Introduction}

The celebrated \emph{incompleteness theorem} by Penrose \cite{rP1965} asserts that, for initial data $(\Sigma, g, k)$ to the Einstein vacuum equations 
$$\mathrm{Ric}(\bar{g}) = 0$$
prescribed on a non-compact initial hypersurface $\Sigma$, the presence of a \emph{trapped surface} in the corresponding maximal development $(\mathcal{M}, \bar{g})$ implies that $(\mathcal{M}, \bar{g})$ is causally geodesically incomplete. As a special case, $(\mathcal{M}, \bar{g})$ is necessarily incomplete if a compact trapped surface is already contained inside the initial hypersurface.

It can be easily verified that Minkowski spacetime $(\mathbb{R}^{3+1}, m)$ is geodesically complete. Thus, in particular, $(\mathbb{R}^{3+1}, m)$ contains no compact trapped surface $S$. The existence of such a surface $S$ inside an initial data set $(\mathbb R^3, g, k)$ can be, therefore, viewed as an indication that $(g,k)$ is far away from the data induced on any spacelike slice of Minkowski spacetime. Indeed, it is easy to check that, when $(g,k)$ is close to the data $(g_0,k_0)$ induced on a spacelike embedding of $\mathbb R^3$ in $(\mathbb R^{3+1},m)$ in a sufficiently strong norm, then $(\mathbb R^3, g,k)$ contains no trapped surfaces. (In the case where $(g,k)$ is close to $(e,0)$, the induced data on the $\{t=0\}$, in a sufficiently strong norm, the global nonlinear stability of Minkowski spacetime, established by Christodoulou--Klainerman \cite{CK}, moreover implies that the corresponding maximal development is also geodesically complete. In principle, one expects that such a global nonlinear stability result also holds for $(g,k)$ close to the induced data on more general spacelike hypersurfaces, for instance by adapting the proof of Lindblad--Rodnianski \cite{hLiR2010}.) \textbf{The purpose of this article is to prove that smallness in a sharp scale-invariant low-regularity norm is already sufficient to rule out the existence of trapped surfaces in the initial hypersurface.} We refer the reader to Section~\ref{sec:background} for connections of this problem to other questions about the Cauchy problem for the Einstein equations in the regime of low regularity.

Our main result is the following:

\begin{theorem}\label{thm:main}
Let $\Omega \subseteq \mathbb R^3$ be a domain and $f:\Omega \rightarrow \mathbb R$ be a smooth function such that the graph $\Sigma_f = \{(t,x^1,x^2,x^3): t=f(x^1,x^2,x^3)\}$ is a uniformly spacelike hypersurface of Minkowski spacetime $(\mathbb R^{3+1},m)$. Then, there exists a constant $\epsilon_0 > 0$ depending only on $\| \partial^2 f\|_{L^\infty (\Omega)}$ and $\inf_\Omega (1-|\rd f|_e)$ such that the following holds: Let $(g,k)$ be a pair of a Riemannian metric and a symmetric covariant $2$-tensor on $\Omega$ satisfying the bound
$$\|g -g_0 \|_{B^{3/2}_{2,1}(\Omega)} + \| k-k_0\|_{B^{1/2}_{2,1}(\Omega)} \leq \ep_0,$$
where $(g_0,k_0)$ are the pullbacks of the Riemannian metric and second fundamental form induced on $\Sigma_f$ by the Minkowski metric $m$. Then, there does not exist a smooth embedded compact trapped $2$-surface $S$ in $(\Omega,g,k)$. 
\end{theorem}

\noindent For the definition of a trapped surface inside $(\Omega,g,k)$, see Section \ref{subsec:Surfaces}. For the definition of the Besov spaces $B^{s}_{2,1}(\Omega)$, see Section \ref{subsec:Besov spaces}. Let us note that the norms appearing in Theorem \ref{thm:main} involve the components of the corresponding tensor fields in standard Cartesian coordinate system on $\mathbb R^3$.

\begin{remark*}
 Notice that Theorem \ref{thm:main} applies to the case of general pairs of $(g,k)$ which do not necessarily have to obey the constraint equations. In the simplest case when $f\equiv 0$, the background tensors $(g_0,k_0)$ reduce to $(e,0)$, where $e$ is the Euclidean metric on $\mathbb R^3$.
\end{remark*}


The Besov norms appearing in Theorem~\ref{thm:main} are invariant under the scaling 
\begin{equation}\label{eq:homothetic.scaling}
(g_{ij}(x),k_{ij}(x)) \rightarrow  ( (g_\lambda)_{ij}, (k_{\lambda})_{ij}) \doteq  (g_{ij}(\lambda x),\lambda k_{ij}(\lambda x)).
\end{equation}
It is clear that the same result as in Theorem~\ref{thm:main} cannot hold for a norm below scaling, e.g.~$H^s(\mathbb R^3)\times H^{s-1}(\mathbb R^3)$ for $s< \f 32$ since (in the $f\equiv 0$ case) one can construct counterexamples simply by rescaling as follows. Take a particular pair $(g,k)$ such that $(g-e,k)$ is compactly supported and such that a smooth embedded compact trapped $2$-surface is present. Consider the rescaled data sets $(g_\lambda, k_\lambda)$ as in  \eqref{eq:homothetic.scaling}. When $\lambda \to \infty$, the $H^s(\mathbb R^3) \times H^{s-1}(\mathbb R^3)$ norm of $(g_\lambda - e, k_\lambda)$ becomes arbitrarily small (for $s<\f 32$), but a compact trapped surface would still be present in all rescaled data sets.

Moreover, in Proposition~\ref{prop:sharp}, we show that the result in Theorem~\ref{thm:main} would still fail if $B^{3/2}_{2,1} \times B^{1/2}_{2,1}$ norm is replaced by $H^{3/2}\times H^{1/2}$. (Note that the proof of Theorem~\ref{thm:main} relies on trace estimates and Sobolev embedding estimates which no longer hold when $B^{3/2}_{2,1} \times B^{1/2}_{2,1}$ norm is replaced by $H^{3/2}\times H^{1/2}$.)
\begin{proposition}\label{prop:sharp}
There exist a sequence $\{(g^{(j)}, k^{(j)}) \}_{j=1}^\infty$ where $g^{(j)}$ are smooth asymptotically flat Riemannian metrics on $\mathbb R^3$, $k^{(j)}$ are smooth and compactly supported symmetric covariant $2$-tensors on $\mathbb R^3$ such that $$\|g^{(j)} -e \|_{H^{3/2}(\mathbb R^3)} + \| k^{(j)} \|_{H^{1/2}(\mathbb R^3)} \leq 2^{-j}$$
(where $e$ is the Euclidean metric on $\mathbb R^3$), but there is a smooth embedded compact trapped $2$-sphere $\Sigma$ in $(\mathbb R^3,g^{(j)},k^{(j)})$ for all $j \in \mathbb N$.
\end{proposition}

Since initial data sets to the evolution problem must satisfy the constraint equations, Proposition~\ref{prop:sharp} may not be fully satisfactory. Instead, one may want to look for counterexamples when the constraint equations are imposed. To study this, we turn to the Einstein--scalar field system in spherical symmetry. First, we show that counterexamples can still be found in this setting. (In particular, this shows that the dominant energy condition would not be sufficient to rule out counterexamples.) 
\begin{proposition}\label{prop:sharp.constraint}
There exists a sequence of initial data sets $\{(g^{(j)}, k^{(j)}; \psi_0^{(j)},\psi_1^{(j)}) \}_{j=1}^\infty$ for the Einstein--scalar field system on $B(0,1)\subset \mathbb R^3$, i.e.~a sequence of Riemannian metrics $g^{(j)}$, symmetric $(0,2)$-tensors $k^{(j)}$ and functions $\psi_0^{(j)}, \psi_1^{(j)} :B(0,1)\rightarrow \mathbb R$ satisfying the constraint equations
\begin{align}
R[g^{(j)}]+(\tr_{g^{(j)}} k^{(j)})^2 - \| k^{(j)} \|^2_{g^{(j)}} &= \|\nabla \psi_0^{(j)} \|^2_{g^{(j)}}+(\psi_1^{(j)})^2, \label{Hamiltonian constraint}\\
\div_{g^{(j)}} k^{(j)} - \nabla (\tr_{g^{(j)}} k^{(j)}) &= \psi_1^{(j)} \nabla \psi_0^{(j)} \label{Momentum constraint},
\end{align}
 such that 
\begin{equation}
\|g^{(j)} -e \|_{H^{3/2}(B(0,1))} + \| k^{(j)} \|_{H^{1/2}(B(0,1))}  \leq 2^{-j}, \quad \forall j \in \mathbb N \label{Smallness counterexample}
\end{equation}
and $\big( B(0,1);g^{(j)},k^{(j)} \big)$ contains a  smooth embedded compact trapped $2$-sphere $\Sigma$ for all $j \in \mathbb N$.
\end{proposition}

However, in the setting of the spherically symmetric Einstein-scalar field system, it is natural to impose an additional smallness assumption on the scalar field in $H^{3/2}\times H^{1/2}$. In this case, it can be shown that spherically symmetric trapped surfaces can be ruled out:
\begin{proposition}\label{prop:SS}
Suppose $(B(0,R), g, k)$ are spherically symmetric taking that form
$$g(\rho,\varphi,\vartheta) = \ud \rho^2 + (r(\rho))^2 \, (\ud \vartheta^2 + \sin^2 \vartheta\, \ud \varphi^2),\quad k(\rho,\varphi,\vartheta) = k_{\rho\rho}(\rho) \ud \rho^2 + k_{\th\th}(\rho)\,  (\ud \th^2 + \sin^2 \vartheta\, \ud \varphi^2)$$
and moreover satisfy the constraints \eqref{Hamiltonian constraint} and \eqref{Momentum constraint} for some smooth and spherically symmetric $\psi_0$, $\psi_1$.

Introduce the corresponding Cartesian coordinates by 
\begin{equation}
    x^1 = \rho \sin \th \cos \varphi,\quad x^2 = \rho \sin \th \sin \varphi,\quad x^3 = \rho \cos \th
\end{equation}
so that in the $(x^1,x^2,x^3)$ coordinate system,
\begin{align*}
    g_{ij} = \f{r^2(|x|)}{|x|^2} \de_{ij} + \Big(1 - \f{r^2(|x|)}{|x|^2}\Big) \f{x^i x^j}{|x|^2}, \quad
    k_{ij} = \f{k_{\th\th}}{|x|^2} \de_{ij} + \Big( k_{\rho\rho} - \f{k_{\th\th}}{|x|^2}\Big)\f{x^i x^j}{|x|^2}.
\end{align*}

Then there exists $\ep_0>0$ (independent of $R$) such that as long as the following smallness condition in the Cartesian  coordinates holds
\begin{equation}\label{eq:SS.smallness}
\|k\|_{H^{\f 12}(B(0,R))}+ \|\psi_0\|_{H^{\f 32}(B(0,R))} + \|\psi_1 \|_{H^{\f 12}(B(0,R))} \leq \ep_0,
\end{equation}
then there does not exist a spherically symmetric smooth embedded compact trapped $2$-surface $S$ in $(B(0,R), g, k)$.
\end{proposition}

In view of Propositions~\ref{prop:sharp.constraint} and \ref{prop:SS}, it is perhaps of interest to understand whether the constraint equations in \emph{vacuum} together with a smallness assumption of $(g,k)\in H^{3/2}\times H^{1/2}$ would be sufficient to rule out trapped surfaces \emph{in the absence of symmetry}.

See Section \ref{sec:Examples} for the proof of Propositions \ref{prop:sharp}--\ref{prop:SS}.

\subsection{Idea of the proof} The proof of the theorem is based on a contradiction argument using a uniform trace theorem. To explain the ideas of the proof, we first focus on the simpler case where $f = 0$, i.e.~we assume that $(g,k) - (e,0)$ is small in $B^{3/2}_{2,1}\times B^{1/2}_{2,1}$. In this setting, assume for the sake of contradiction that there is a compact trapped surface $S$. The following are the main steps of the argument.
\begin{enumerate}
\item Let $S^+$ be the convex hull of $S$ intersected with $S$, and let $H_0$ be the Euclidean mean curvature. By convexity, $H_0 \geq 0$ on $S^+$. Moreover, it is well-known that the Willmore energy has the following lower bound: 
$$\int_{S^+} H_0^2 \, \mathrm{dVol}_{S,e} \geq 16\pi.$$ 
(This can be derived by noting that (a) the standard Gauss map $\hat{n}: S^+ \to (\mathbb S^2,\slashed{g}_0)$ covers the whole of $\mathbb S^2$ and thus, $\int_{S^+} K_{\slashed{g}_0} \, \mathrm{dVol}_{S,e} \geq 4\pi$ and (b) $H_0^2 \geq 4 K_{\slashed{g}_0}$ by the AM-GM inequality.)
\item The key ingredient we establish is a uniform trace estimate for convex hypersurfaces in $\mathbb R^3$: 
\begin{equation}\label{eq:intro.trace}
\int_{\mathcal S} |\phi|^2\, \mathrm{dVol}_{\mathcal S,e} \ls \| \phi \|_{B^{1/2}_{2,1}}^2.
\end{equation}
where the implicit constant \emph{independent} of the surface $\mathcal S$, as long as it is convex. (It is easy to see that a trace estimate cannot hold uniformly for all hypersurfaces without the convexity assumption.)
\item Since $S^+$ is convex and $(g,k) - (e,0)$ is small, say of size $O(\ep)$, in $B^{3/2}_{2,1}\times B^{1/2}_{2,1}$, we can apply the uniform trace estimate in Step~2 so as to obtain 
$$\int_{S^+} |k|^2 \, \mathrm{dVol}_{S,e},\int_{S^+} |\rd (g - g_0)|^2 \, \mathrm{dVol}_{S,e} \ls \ep.$$
\item Since $S^+$ is trapped, $\mathrm{tr} k + H <0$. In particular, we have $0\leq H_0 < H_0 - H - \mathrm{tr} k$. Moreover, using $H_0 \geq 0$ and that $\| g-g_0\|_{L^\infty} \ls \ep$ (by Sobolev embedding $B^{3/2}_{2,1} \hookrightarrow L^\infty$), we have the pointwise bound $|H_0 - H|\ls |\rd(g-g_0)| + \ep H_0$ (see computations in Lemma~\ref{lem:Comparison geometries}).
Hence, after applying the estimates in Step~3, we obtain
\begin{equation*}
\begin{split}
&\: \int_{S^+} |H_0|^2 \, \mathrm{dVol}_{S,e} \ls \int_{S^+} |H - H_0|^2 \, \mathrm{dVol}_{S,e} + \int_{S^+} |k|^2 \, \mathrm{dVol}_{S,e} \\
 \ls &\: \int_{S^+} \Big(  |k|^2 + |\rd(g-g_0)|^2 \Big) \, \mathrm{dVol}_{S,e} + \ep \int_{S^+} |H_0|^2 \, \mathrm{dVol}_{S,e} \ls \ep + \ep \int_{S^+} |H_0|^2 \, \mathrm{dVol}_{S,e}.
\end{split}
\end{equation*}
For $\ep>0$ sufficiently small, we have $\int_{S^+} |H_0|^2 \, \mathrm{dVol}_{S,e} \ls \ep$, contradicting the lower bound in Step~1.
\end{enumerate}

It remains to explain the uniform trace estimate \eqref{eq:intro.trace} used in Step~2. Partition $\mathbb S^2 \doteq \{x\in \mathbb R^3: \| x\|=1\}$ into $6$ pieces $\mathbb S^2 = \bigcup_{i=1}^3 \mathcal E_i \cup \bigcup_{i=1}^3 \mathcal W_i$, where $\mathcal E_i\doteq \{ x\in \mathbb S^2: x_i \geq \f 12\}$, $\mathcal W_i \doteq \{ x\in \mathbb S^2: x_i \leq - \f 12\}$. This induces a partition $\mathcal S = \bigcup_{i=1}^3 \hat{n}^{-1}(\mathcal E_i) \cup \bigcup_{i=1}^3 \hat{n}^{-1}(\mathcal W_i)$. Convexity implies that each of $\hat{n}^{-1}(\mathcal N_i)$, $\hat{n}^{-1}(\mathcal S_i)$ can be written as a graph, and we can adapt the standard proof of trace estimates.

In the more general case where $f \not\equiv 0$, we need a (spacetime) notion of the null convex hull of a $2$-surface, which is defined to be the intersection of null half-spaces containing the surface. In this case, instead of the convex hull, we consider the intersection of the boundary of the null convex intersected with the surface; and the quantity we consider in place of $H_0$ is the Minkowski null expansion $\tr \chi_0$. (Notice that in the case of $f = 0$, $\tr \chi_0 = H_0$.) It turns out that suitable analogues of the key properties (1) and (2) above still hold in the more general setting, using slightly more involved arguments; see Proposition~\ref{prop:S+}, Lemma~\ref{lem:patches} and Lemma~\ref{lem:trace}.

\subsection*{Acknowledgements} We would like to express our gratitude to Otis Chodosh for many helpful discussions and for numerous insightful comments during the early stages of this work. We would also like to thank Or Hershkovits and Rafe Mazzeo for fruitful discussions. We thank an anonymous referee for many useful suggestions to improve the exposition. We also thank Chao Wu for pointing out a number of omissions in the earlier version of the paper. The first author acknowledges the support of a Terman fellowship and the NSF grant DMS-2005435. The second author acknowledges the support of the Clay Mathematics Institute while this work was being completed.

\section{Motivation: Cauchy problem and optimal low-regularity well-posedness for the Einstein vacuum equations}\label{sec:background}

The main motivation for Theorem~\ref{thm:main} comes from the following fundamental question: 

\medskip

\begin{center}
    \textbf{What is the threshold of well-posedness for the Einstein vacuum equations (or appropriate Einstein-matter systems) when considering low-regularity initial data?}
\end{center}

\medskip

For many evolutionary partial differential equations, low-regularity well-posedness problems are often important for understanding singularity formation. In the setting of the Einstein equations, a prominent example can be found in the works of Christodoulou \cite{dC1999}, in which the resolution of the weak cosmic censorship conjecture for the Einstein--scalar field system in spherical symmetry relied on Christodoulou's sharp BV well-posedness result \cite{dC1993}.

The question of optimal low-regularity well-posedness can be formulated in \emph{local} or \emph{global} terms:

\begin{problem}\label{prob:main}
For $s \geq \f 32$, let $X = H^s(\mathbb R^3) \times H^{s-1}(\mathbb R^3)$ (or a suitable weighted version or Besov replacement). Does there exist a sequence of initial data sets $\{(\mathbb R^3, g_i, k_i)\}_{i=1}^\infty$ to the Einstein vacuum equations such that 
$$\|( g_i - e, k_i )\|_{X} \leq 2^{-i},$$
and for which:
\begin{enumerate}
    \item The solution does not remain of size $O(2^{-i})$ in the norm $\|\cdot\|_X$ ``up to time $O(1)$''?
    \item The corresponding maximal globally hyperbolic development is future causally geodesically incomplete?
\end{enumerate}
\end{problem}

 Part (1) of Problem~\ref{prob:main} probes the regularity threshold below which the local existence of solutions ceases to hold. The best known result in this direction is the celebrated bounded $L^2$ curvature theorem of Klainerman--Rodnianski--Szeftel \cite{sKiRjS2015}, which established that (modulo technical assumptions) solutions to the Einstein vacuum equations remain under control up to time $O(1)$ if the initial data are small in $X = H^2(\mathbb R^3) \times H^{1}(\mathbb R^3)$. As pointed out in \cite{sKiRjS2015}, the $L^2$ bound of curvature is crucially used in the proof to derive a lower bound on the radius of injectivity of null hypersurfaces, and it is therefore unclear whether the solutions can be controlled below this regularity; see also \cite{sKiR2005, sKiR2008}.

 Part (2) of Problem~\ref{prob:main} is related to the question of the stability of Minkowski spacetime in the roughest possible setting. This is closely connected to the question of trapped surface formation, since the emergence of a trapped surface implies that the solution is geodesically incomplete due to Penrose's incompleteness theorem. It is known by Christodoulou's monumental work \cite{Chr} that trapped surfaces can form dynamically from initial data which are free of trapped surfaces (and in fact are arbitrarily far from having trapped surfaces). The result in \cite{Chr} requires that the initial data are large in $H^1$. In a subsequent work \cite{AL}, An--Luk showed that largeness in $H^{3/2}$ is already sufficient to guarantee that trapped surfaces form dynamically. 

Our main result (Theorem~\ref{thm:main}) only concerns the existence of trapped surfaces within initial data sets and, thus, does not directly address the evolution problem. Our theorem shows that if the Cauchy data to the Einstein vacuum equations contain a trapped surface, then the data cannot be small in the scale-invariant $B^{3/2}_{2,1}$ norm. Hence, if a trapped surface is to emerge in evolution, then the $B^{3/2}_{2,1}$ norm of the data induced on spacelike slices of the spacetime cannot remain small. Together with the possibility of inflation for the $H^s$ norm along the evolution when $s<2$, one is naturally led to the following question, which can be viewed as a reformulation of Problem \ref{prob:main} in the context of trapped surface formation:

\begin{problem}
For $s\in [\f 32,2)$, let $X = H^s(\mathbb R^3) \times H^{s-1}(\mathbb R^3)$ (or a suitable Besov replacement). Does there exist a sequence of initial data sets $\{(\mathbb R^3, g_i, k_i)\}_{i=1}^\infty$ to the Einstein vacuum equations such that 
$$\|( g_i - e, k_i )\|_{X} \leq 2^{-i},$$
but for which the corresponding maximal globally hyperbolic future development contains an embedded compact trapped surface?
\end{problem}

\section{Notations} \label{sec:Notations}

In this section, we will introduce the various notational conventions that we will adopt throughout the paper.

\subsection{Special subsets of Minkowski spacetime}
We will denote with $m$ the Minkowski metric on $\mathbb R^{3+1}$, which, in the standard Cartesian coordinates $(t,x^1,x^2,x^3)$, takes the form
\[
m = -dt^2 +(dx^1)^2+(dx^2)^2+(dx^3)^2.
\]
We will also denote with $e$ the Euclidean metric on $\mathbb R^3$:
 \[
 e = (dx^1)^2+(dx^2)^2+(dx^3)^2.
 \]
 We will frequently identify $\mathbb R^3$ with $\{ t = 0 \} \subset \mathbb R^{3+1}$. We will also identify $\mathbb S^2$ with the coordinate sphere $\{ \sum_{i=1}^3 (x^i)^2=1\}$ in $\mathbb R^3$. In what follows, lower case Latin indices run through $i,j=1,2,3$.
 
 \begin{definition}\label{def:Null vector sphere}
We will define for any $\omega = ( \omega^1,\omega^2,\omega^3) \in \mathbb S^2\subset  \mathbb R^3$ the vector
\[
L_\omega \doteq (1, \omega^1,\omega^2,\omega^3) \in \mathbb R^{3+1}.
\]
 We will also define for any $\omega \in \mathbb S^2$ and $u\in \mathbb R$ the half space
\begin{equation}\label{Null halfspace}
W_{\omega,u} \doteq \big\{ (t,x) \in \mathbb R^{3+1}: t- \langle x, \omega \rangle_e \ge u \big\}.
\end{equation}
\end{definition}
\begin{remark*}
 Note that $L_\omega$ is future directed and null with respect to $m$. Moreover, the boundary 
\[
\Pi_{\omega,u} \doteq  \partial W_{\omega,u} =  \big\{ (t,x) \in \mathbb R^{3+1}: t- \langle x, \omega \rangle_e =u \big\}
\]
 is a null hyperplane of $(\mathbb R^{3+1},m)$ whose normal vector at every point is (parallel to) $L_\omega$.
 \end{remark*}
 
 \subsection{Spacelike hypersurfaces in $\mathbb R^{3+1}$}
 Throughout this paper, we will frequently consider spacelike hypersurfaces of $\mathbb (\mathbb R^{3+1},m)$ which can be expressed as a graph of a given smooth function $f:\Omega \rightarrow \mathbb R$ over a domain $\Omega \subseteq \mathbb R^3$:
 \[
 \Sigma_f \doteq \big\{ (t,x): \, x\in \Omega, \, t=f(x) \big\}.
 \]
 Note that $\Sigma_f$ is spacelike if and only if 
 \begin{equation}\label{Spacelike gradient}
 |\rd f|_{e} <1 \quad \text{everywhere on} \quad \Omega,
 \end{equation}
where, from now on, we use $\rd f$ to denote the Euclidean gradient of $f$ and $|\rd f |_e = (\sum_{i=1}^3 |\rd_i f|^2)^{\f 12}$. In what follows, we will only consider functions $f$ satisfying the uniform bound $\inf_\Omega (1-|\rd f|_e ) >0$.
 
For a spacelike $\Sigma=\Sigma_f$ as above, we will denote with $n_\Sigma$ the future-directed unit timelike normal of $\Sigma$ (with respect to $m$), i.e.~$n_\Sigma = \f{1}{\sqrt{1-|\rd f|_e^2}}(\rd_t + \de^{ij} \rd_i f \rd_j)$. We will also denote with $(g_0,k_0)$ the  Riemannian metric and second fundamental form induced on $\Sigma=\Sigma_f$ by $m$, i.e.
\[
g_0(X,Y) = m(X,Y) \quad \text{and} \quad k_0(X,Y) = \langle D_X n_\Sigma, Y \rangle_m  \quad \text{for any } \hphantom{n} X,Y \hphantom{n} \text{ tangent to } \hphantom{n} \Sigma,
\]
where $D$ is the flat connection on $\mathbb R^{3+1}$. 

We will frequently identify a hypersurface $\Sigma_f$ with the domain of support of $f$ in $\mathbb{R}^3$ (via the map $f$), and denote simply with $g_0, k_0$ the respective pullbacks $f_* g_0, f_* k_0$. In a system of Cartesian coordinates $(x^1, x^2, x^3)$ on $\Omega$, the tensors $g_0$ and $k_0$ take the form
\begin{align}
(g_0)_{ij} & = \delta_{ij}-\partial_{i} f \cdot \partial_j f, \label{eq:Flat metric}\\
(k_0)_{ij} & = \frac{\partial^2_{ij} f}{\sqrt{1-|\rd f|^2_e}}.\label{eq:Flat k}
\end{align}

\subsection{The geometry of embedded $2$-surfaces}\label{subsec:Surfaces} Let $S \hookrightarrow \mathbb R^3$ be an embedded, connected, smooth $2$-surface. Such a surface is necessarily orientable and separates $\mathbb R^3$ into two components, a compact one (which we will denote with $\mathcal{K}_{int}$) and a non-compact one (which we will call $\mathcal{K}_{ext}$); see \cite{hS1969}.
\begin{definition}
Let $S \hookrightarrow \mathbb R^3$ be a closed embedded surface as above and let $g$ be a Riemannian metric defined in an open neighborhood of $S$ in $\mathbb R^3$.
\begin{enumerate}
\item For all $x\in S$, we will denote with $N_g(x) \in T_x \mathbb{R}^3$ the unit normal to $S$ at $x$ with respect to the metric $g$ pointing in the direction of $\mathcal K_{ext}$. 
\item We will denote with $\sg$ the induced metric on $S$ by $g$, i.e.
\[
\sg (X,Y)=g(X,Y) \quad \text{for all} \quad X,Y \quad \text{tangent to} \quad S.
\]
\item We will denote with $h_g$ the second fundamental form of $S$ associated to $g$ and $N_g$, i.e.
\[
h_g(X,Y) =  \langle \nabla_X N_g,Y \rangle_g \quad \text{for any} \quad X,Y \quad \text{tangent to} \quad S,
\]
where $\nabla$ is the connection of $(\mathbb R^3,g)$. 
\end{enumerate}
\end{definition}

We will adopt the following definition for a trapped surface in $(\mathbb R^3, g,k)$:
\begin{definition}[Trapped surfaces]
Let $g$ be Riemannian metric on $\Omega \subseteq \mathbb R^3$ and $k$ a symmetric covariant $2$-tensor on $\Omega$. Let also $S \hookrightarrow \Omega$ be a \textbf{compact}, embedded and smooth $2$-surface. We will say that $S$ is a \textbf{trapped} surface in $( \Omega,g,k)$ if, at every point on $S$, the $(0,2)$-tensor $k$ (restricted to $S$) and the second fundamental form $h_g$ of $S$ satisfy
$$\mathrm{tr}_{\sg} (k+h_g) <0,\quad \mathrm{tr}_{\sg} (k-h_g) <0.$$
\end{definition}

\subsection{Null convex hulls of $2$-surfaces in $(\mathbb R^{3+1},m)$}\label{subsec:Null convex hulls}

Let $S$ be a smooth, connected, closed and  embedded $2$ surface contained inside a domain $\Omega \subseteq \mathbb R^3$ and let $f : \Omega \rightarrow \mathbb R$ be a smooth function satisfying the gradient bound \eqref{Spacelike gradient} (so that $\Sigma_f$ is a spacelike hypersurface of $(\mathbb R^{3+1}, m)$). Define also $\overline{f}:\Omega \to \mathbb R^{3+1}$ by $\overline{f}(x) = (f(x),x)$.

\begin{definition}\label{def:Normalized null normal}
 For any $p\in S$, we will define $L^{\overline{f}(S)}[p]$ to be the outgoing null normal to $\overline{f}(S)$ at $p$ with respect to $m$, normalized so that $\langle L^{\overline{f}(S)}[p], \partial_t \rangle_m = -1$; that is to say, $L^{\overline{f}(S)}[p]$ is the unique vector in $T_p \mathbb R^{3+1}$ with
\begin{equation}\label{Normal fS}
\langle L^{\overline{f}(S)}[p],X \rangle_m=0 \quad \text{for all} \quad X\in T_p \overline{f}(S)
\end{equation}
and which is of the form 
\[
L^{\overline{f}(S)}[p] \doteq (1, v^1, v^2, v^3) 
\] 
with $|v|^2_e =1$ and $v$ pointing to $\mathcal{K}_{ext}$. We will also define $\Pi^{\overline{f}(S)}[p]$ to be the null plane containing $p$ with generator $L^{\overline{f}(S)}[p]$, i.e.
\[
\Pi^{\overline{f}(S)}[p] \doteq \big\{ z \in \mathbb R^{3+1}:\, \langle z-p, L^{\overline{f}(S)}[p] \rangle_m = 0 \big\}.
\]
\end{definition}
\begin{remark*}
Note that, in view of \eqref{Normal fS}, $\Pi^{\overline{f}(S)}[p] $ is necessarily tangent to $S$ at $p$.
\end{remark*}

\begin{figure}[t]
\scriptsize
\begingroup%
  \makeatletter%
  \providecommand\color[2][]{%
    \errmessage{(Inkscape) Color is used for the text in Inkscape, but the package 'color.sty' is not loaded}%
    \renewcommand\color[2][]{}%
  }%
  \providecommand\transparent[1]{%
    \errmessage{(Inkscape) Transparency is used (non-zero) for the text in Inkscape, but the package 'transparent.sty' is not loaded}%
    \renewcommand\transparent[1]{}%
  }%
  \providecommand\rotatebox[2]{#2}%
  \newcommand*\fsize{\dimexpr\f@size pt\relax}%
  \newcommand*\lineheight[1]{\fontsize{\fsize}{#1\fsize}\selectfont}%
  \ifx\svgwidth\undefined%
    \setlength{\unitlength}{225bp}%
    \ifx\svgscale\undefined%
      \relax%
    \else%
      \setlength{\unitlength}{\unitlength * \real{\svgscale}}%
    \fi%
  \else%
    \setlength{\unitlength}{\svgwidth}%
  \fi%
  \global\let\svgwidth\undefined%
  \global\let\svgscale\undefined%
  \makeatother%
  \begin{picture}(1,0.5)%
    \lineheight{1}%
    \setlength\tabcolsep{0pt}%
    \put(0,0){\includegraphics[width=\unitlength,page=1]{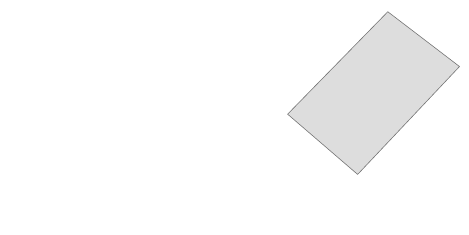}}%
    \put(0.85610428,0.17254221){\makebox(0,0)[lt]{\lineheight{1.25}\smash{\begin{tabular}[t]{l}$\Pi^{\overline f(S)}[p]$\end{tabular}}}}%
    \put(0.68606193,0.20537218){\makebox(0,0)[lt]{\lineheight{1.25}\smash{\begin{tabular}[t]{l}$p$\end{tabular}}}}%
    \put(0.81738185,0.30554563){\makebox(0,0)[lt]{\lineheight{1.25}\smash{\begin{tabular}[t]{l}$L^{\overline f(S)}[p]$\end{tabular}}}}%
    \put(0,0){\includegraphics[width=\unitlength,page=2]{L_and_Pi.pdf}}%
    \put(0.37014146,0.10992629){\makebox(0,0)[lt]{\lineheight{1.25}\smash{\begin{tabular}[t]{l}$\overline{f}(S)$\end{tabular}}}}%
    \put(0,0){\includegraphics[width=\unitlength,page=3]{L_and_Pi.pdf}}%
  \end{picture}%
\endgroup%

\caption{Schematic depiction of the null outer normal vector $L^{\overline{f}(S)}[p]$ and the null plane $\Pi^{\overline{f}(S)}[p]$ associated to a point $p$ on the spacelike surface $\overline f(S) \subset \mathbb R^{3+1}$.}
\end{figure}

\begin{definition}\label{def:Null expansion}
Let $S \hookrightarrow  \Omega \subseteq \mathbb R^3$ and $f$ be as above. We will define the flat null expansion of $S$ by the relation
\[
\tr \chi_0 \doteq \tr_{\sg_0} (D L^{\overline{f}(S)}),
\]
where $\sg_0$ is the Riemannian metric induced on $\overline{f}(S)$ by $m$ and the $(0,2)$-tensor $D L^{\overline{f}(S)}$ on the surface $\overline{f}(S)$ (which we will call the null second fundamental form) is defined by 
\[
D L^{\overline{f}(S)}(X,Y) = \langle D_X L^{\overline{f}(S)}, Y \rangle_m \quad \text{ for any } X,Y \text{ tangent to } \overline{f}(S)
\]
(where $D$ is the flat connection on $\mathbb R^{3+1}$).
\end{definition}

\begin{remark*}
Note that $L^{\overline{f}(S)}=\zeta_S \cdot (n_{\Sigma_f}+\bar f^* N_{g_0})$ with $\bar f^* N_{g_0}$ being the pushforward of $N_{g_0}$ via $\bar f$ (hence, by the definition of $g_0$, $\langle \bar f^* N_{g_0}, \bar f^* N_{g_0} \rangle_m=1$) and $$\zeta_S \doteq  \f{\sqrt{1-|\partial f|_e^2}}{1+N_{g_0}(f) \sqrt{1-|\partial f|_e^2}}$$ (note that since $|N_{g_0}(f)| \leq |\rd f|_e(1-|\rd f|_e^2)^{-\f 12}$, we have $\f12 (1-|\rd f|_e^2)^{\f12} \le \zeta_S \le 2 (1-|\rd f|_e)^{-\f12}$). Thus,
\begin{equation}\label{eq:trch0}
\tr\chi_0 = \zeta_S \,\tr_{\sg_0}(k_0 + h_{g_0}).
\end{equation}
\end{remark*}

\begin{definition}[The subset $S^+ \subseteq S$]\label{def:Null convex hull}
Let $S \hookrightarrow  \Omega \subseteq \mathbb R^3$ and $f$ be as above. We will define the \textbf{null convex hull} of $S$ to be the subset of $\mathbb R^{3+1}$ consisting of the intersection of all half-spaces of the form \eqref{Null halfspace} containing $\overline{f}(S)$, i.e.
\begin{equation}\label{Null convex hull}
\mathbb{K}_+[S] \doteq \bigcap \big\{ W_{\omega, u}:\, \overline{f}(S) \subset W_{\omega, u} \big\}.
\end{equation}
We will then set
\begin{equation}
S^+ \doteq \overline{f}^{-1} \big( \overline{f}(S) \cap \partial \mathbb{K}_+[S]\big) \subseteq S.
\end{equation}
\end{definition}

\begin{figure}[t]
\scriptsize
\begingroup%
  \makeatletter%
  \providecommand\color[2][]{%
    \errmessage{(Inkscape) Color is used for the text in Inkscape, but the package 'color.sty' is not loaded}%
    \renewcommand\color[2][]{}%
  }%
  \providecommand\transparent[1]{%
    \errmessage{(Inkscape) Transparency is used (non-zero) for the text in Inkscape, but the package 'transparent.sty' is not loaded}%
    \renewcommand\transparent[1]{}%
  }%
  \providecommand\rotatebox[2]{#2}%
  \newcommand*\fsize{\dimexpr\f@size pt\relax}%
  \newcommand*\lineheight[1]{\fontsize{\fsize}{#1\fsize}\selectfont}%
  \ifx\svgwidth\undefined%
    \setlength{\unitlength}{225bp}%
    \ifx\svgscale\undefined%
      \relax%
    \else%
      \setlength{\unitlength}{\unitlength * \real{\svgscale}}%
    \fi%
  \else%
    \setlength{\unitlength}{\svgwidth}%
  \fi%
  \global\let\svgwidth\undefined%
  \global\let\svgscale\undefined%
  \makeatother%
  \begin{picture}(1,0.53333333)%
    \lineheight{1}%
    \setlength\tabcolsep{0pt}%
    \put(0,0){\includegraphics[width=\unitlength,page=1]{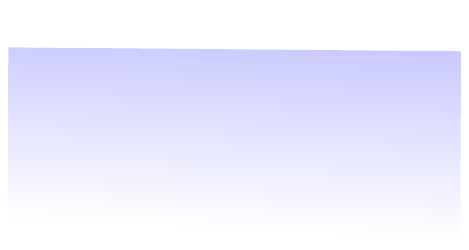}}%
    \put(0.48571499,0.15620969){\makebox(0,0)[lt]{\lineheight{1.25}\smash{\begin{tabular}[t]{l}$\overline{f}(S)$\end{tabular}}}}%
    \put(0.12121831,0.27153548){\makebox(0,0)[lt]{\lineheight{1.25}\smash{\begin{tabular}[t]{l}${\color{red}\overline{f}(S^+)}$\end{tabular}}}}%
    \put(0.02862099,0.3489805){\makebox(0,0)[lt]{\lineheight{1.25}\smash{\begin{tabular}[t]{l}${\color{blue}W_{\omega,u}}$\end{tabular}}}}%
    \put(0.36028776,0.46262266){\makebox(0,0)[lt]{\lineheight{1.25}\smash{\begin{tabular}[t]{l}${\color{blue}\partial W_{\omega,u}}$\end{tabular}}}}%
    \put(0.4714045,0.07539751){\makebox(0,0)[lt]{\lineheight{1.25}\smash{\begin{tabular}[t]{l}${\color{olive}\mathbb K_+[S]}$\end{tabular}}}}%
    \put(0,0){\includegraphics[width=\unitlength,page=2]{Convex_hull.pdf}}%
  \end{picture}%
\endgroup%

\caption{The above figure is a schematic depiction of a slice of the form $\{t=t_0\}$ in the special case when $\overline f(S) \subset \{t=t_0\}$ (i.e.~when $f=\mathrm{const}$ along $S$). In this case, the set $\mathbb K_+[S] \cap \{t=t_0\}$ (depicted in brown), which is simply the intersection of all half-spaces $W_{\omega,u}\cap \{t=t_0\}$ containing $f(S)$ (a typical such half-space is depicted in blue), reduces to the convex hull of $\overline f(S)$ inside $\{t=t_0\}$. The set $\overline f(S^+)$ (depicted with the red dotted line) is simply the set of points on $\overline f(S)$ which also lie on the convex boundary $\partial \mathbb K^+[S] \cap \{t=t_0\}$. }
\end{figure}

\begin{remark*}
In the trivial case $\Omega = \mathbb R^3$ and $f=0$, the hypersurface $\Sigma_f = \Sigma_0$ is simply the hyperplane $\{ t = 0\}$. In that case, it can be easily verified that $\mathbb{K}_+[S] \cap \Sigma_0$ is simply the convex hull of $S \subset \mathbb R^3$ and thus $S^+$ is contained in the boundary of a convex body. 

For $S$ as in Definition \ref{def:Null convex hull}, the set $S^+ \subseteq S$ is always  non-empty; see Proposition \ref{prop:S+}.
\end{remark*}

We can readily infer the following properties for the set $S^+$ (which are similar to the properties of the boundary of the convex hull of a surface in $\mathbb R^3$):
\begin{lemma}\label{lem:Positivity of DL}
Let $S$ and $S^+ \subseteq S$ be as in Definition \ref{def:Null convex hull}. For any point $p\in S^+$, $\overline{f}(S)$ lies on one side of the null hyperplane $\Pi^{\overline{f}(S)}[p]$.  Moreover, the tensor $DL^{\overline{f}(S)}$ is semi-positive definite on $\overline{f}(S^+)$.
\end{lemma}
\begin{proof}
Let $p$ be a point in $S^+$, and set $q=\overline{f}(p)$ to be the corresponding point on $\overline{f}(S^+) \subset \Sigma_f \subset \mathbb R^{3+1}$. Since $q\in \overline{f}(S^+) \subset \partial \mathbb K_+[S]$, there exists a sequence of points $q_n\in \mathbb R^{3+1} \setminus \mathbb K_+[S]$ with $q_n \xrightarrow{n\rightarrow \infty} q$. The definition of $\mathbb K_+[S]$ then implies that, for each $n$, there exists a half space of the form $W_{\omega_n, u_n}$ such that $\overline{f}(S)\subset W_{\omega_n, u_n}$ and $q_n \notin W_{\omega_n, u_n}$. After possibly restricting to a subsequence, the sets $W_{\omega_n, u_n}$ converge to a half space $W_{\omega_\infty, u_\infty}$ such that $q\in \partial W_{\omega_\infty, u_\infty}$. Then, $T_q \overline{f}(S) \subset \partial W_{\omega_\infty, u_\infty}$ since, if $T_q \overline{f}(S)$ was transversal to $\partial W_{\omega_\infty, u_\infty}$, the null hyperplanes $\partial W_{\omega_n, u_m}$ would have to intersect $\overline{f}(S)$ transversally for $n$ large enough. Therefore, the null generator $L_{\omega_\infty}$ of $\partial W_{\omega_\infty, u_\infty}$ is normal to $T_q \overline{f}(S)$. The fact that $L_{\omega_\infty}$ is equal to the outgoing normal $L^{\overline{f}(S)}$ (and not the ingoing one) follows from the observation that $\overline{f}(S)$ is contained in the \emph{future} of $\partial W_{\omega_\infty, u_\infty}$.

We will now establish the non-negativity of the null second fundamental form on $\overline{f}(S^+)$. For $p\in S^+$ as above, let $(x^1, x^2, x^3)$ be a Cartesian coordinate system in $\mathbb R^3$ centered around $p$ such that $S$ is of the form $x^3= \f 12 M_{AB}\bar{x}^A \bar{x}^B+O(|\bar{x}|^3)$ (with $M_{AB}$ being constants, $A,B$ taking the values $1,2$ and $\bar{x} = (x^1, x^2)$) and $\langle L^{\overline{f}(S)}, \partial_3 \rangle_m >0$. Then, the surface $\overline{f}(S)$ can be locally expressed around $q$ as
\[
\overline{f}(S) = \big\{(t, x^1, x^2, x^3):\, x^3= \f12 M_{AB}x^A x^B+O(|\bar{x}|^3), \, t=f(x^1, x^2, x^3) \big\}.
\]
Thus, setting $\bar{v} \doteq (\partial_1 f(0), \partial_2 f(0))$, we can express the Taylor expansion of $L^{\overline{f}(S)}$ (see \eqref{Normal fS}) around $p$ as follows:
\begin{align*}
L^{\overline{f}(S)}(\bar{x}^1, \bar{x}^2) = & (1, \bar{v}^1, \bar{v}^2, \sqrt{1-|\bar{v}|^2}) \\
& + \big(0,\, \partial^2_{1A}f(0) +(\partial_3 f(0)- \sqrt{1-|\bar{v}|^2}) M_{1A} , \, \partial^2_{2A}f(0) + (\partial_3 f(0)- \sqrt{1-|\bar{v}|^2})  M_{2A} , \, 0 \big) \cdot \bar{x}^A \\
& + O(|\bar{x}|^2).
\end{align*}
 Since the tangent space of $\overline{f}(S)$ at $q$ is spanned by $\partial_A + \partial_A f(0) \partial_t$, $A=1,2$, we can express the tensor $DL^{\overline{f}(S)}[p]$ as
\begin{equation}\label{DL first expression}
(DL^{\overline{f}(S)}[p])_{AB} = \partial^2_{AB}f(0)+ (\partial_3 f(0)- \sqrt{1-|\bar{v}|^2})  M_{AB}.
\end{equation}

In view of the fact that $p\in S^+$, we know that $S$ is contained on the future half-space defined by the null hyperplane $\Pi^{\overline{f}(S)}[p]$. In particular, the affine function
\[
\phi (t, x) = \langle (t,x),\, L^{\overline{f}(S)}[p] \rangle_m +f(0)  
\]
satisfies 
\begin{equation}\label{positive phi}
\phi \le 0 \text{ on } \overline{f}(S) \text{ and } \phi(q)=\phi(f(0),0,0,0)=0.
\end{equation}
Expressing $\phi|_{\overline{f}(S)}$ in terms of the coordinates $\bar{x}^A$, $A=1,2$ using the local expression for $\overline{f}(S)$ and the fact that $L^{\overline{f}(S)}[p]=(1, \bar{v}^1, \bar{v}^2, \sqrt{1-|\bar{v}|^2})$, we get
\[
\phi|_{\overline{f}(S)}(\bar{x}^1, \bar{x}^2) =  - \f12 \big( \partial^2_{AB} f(0) \bar{x}^A \bar{x}^B + (\partial_3 f(0)- \sqrt{1-|\bar{v}|^2})  M_{AB}\bar{x}^A \bar{x}^B \big) +O(|\bar{x}|^3).
\]
Thus, since $\partial^2_{AB} \phi(0) $ is semi-negative definite (in view of \eqref{positive phi}), we infer that \eqref{DL first expression} is semi-positive definite.
\end{proof}

\subsection{Function spaces for scalars and tensors}\label{subsec:Besov spaces}
In this section, we will introduce the function spaces that will be used to measure the ``size'' of various tensors on (subsets of) $\mathbb R^3$.

\begin{definition}[Besov spaces on $\mathbb{R}^3$]
Let $\eta:\mathbb R^3 \to [0,1]$ be a radial smooth function such that \[
\eta(\xi) = \begin{cases}
1 & \mbox{for $|\xi|\leq 1$}, \\
0 & \mbox{for $|\xi|\geq 2$}.
\end{cases}
\] 
For every Schwartz function $\phi:\mathbb R^3 \to \mathbb R$, the Littlewood--Paley projections $P_k \phi$, $k \ge 0$, will be defined by
$$P_0 \phi = \mathcal F^{-1}(\eta(\xi) \mathcal F\phi),\quad P_k \phi = \mathcal F^{-1}\Big( (\eta(2^{-k}\xi) - \eta(2^{-k+1}\xi))\mathcal F \phi\Big),\, k\geq 1,$$
where $\mathcal F$ denotes the Fourier transform.
\begin{enumerate}
\item We will define the the Besov spaces $B^s_{p,q}(\mathbb R^3)$ (for $s \geq 0$, $p,q \in [1,\infty)$) as the completion of the space of Schwartz functions $\phi:\mathbb R^3 \rightarrow \mathbb R$ under the following norm:
$$\| \phi \|_{B^s_{p,q}(\mathbb R^3)} \doteq \Big(\sum_{k\geq 0} 2^{qsk} \| P_k \phi \|_{L^p(\mathbb R^3)}^q \Big)^{1/q}.$$
We will also set
\[
H^s(\mathbb R^3)\doteq B^s_{2,2}(\mathbb R^3).
\]
\item For a covariant $2$-tensor $\phi$ on $\mathbb R^3$, we will define $\|\phi\|_{B^s_{p,q}(\mathbb R^3)}$ (and $\|\phi \|_{H^s(\mathbb R^3)}\doteq \|\phi \|_{B^s_{p,2}(\mathbb R^3)}$) in terms of the components of $\phi$ in the (fixed) Cartesian coordinate system, i.e.
$$\| \phi \|_{B^s_{p,q}(\mathbb R^3)} \doteq \sum_{i,j=1}^3 \| \phi_{ij} \|_{B^s_{p,q}(\mathbb R^3)}.$$
\end{enumerate}
\end{definition}

In the case of a domain $\Omega \subset \mathbb R^3$, the Besov space $B^s_{p,q}(\Omega)$ will be defined to consist of the restriction of $B^s_{p,q}(\mathbb R^3)$ functions to $\Omega$:

\begin{definition}[Besov spaces on $\Omega \subset \mathbb R^3$]\label{def:Hs.local}
Let $\Omega \subset \mathbb R^3$ be an open set and $s \geq 0, p,q\in [1,+\infty)$. Let $\phi: \Omega\to \mathbb R$ be a measurable function. Define
$$E_{B^s_{p,q},\phi}\doteq \{\overline{\phi}: \mathbb R^3\to \mathbb R: \overline{\phi}_{|\Omega} = \phi,\, \overline{\phi} \in B^s_{p,q}(\mathbb R^3)\}.$$
We say that $\phi \in B^s_{p,q}(\Omega)$ if $E_{B^s_{p,q},\phi} \neq \emptyset$. We will define:
$$\|\phi\|_{B^s_{p,q}(\Omega)} \doteq \inf_{\overline{\phi} \in E_{B^s_{p,q},\phi}} \|\overline{\phi}\|_{B^s_{p,q}(\mathbb R^3)}.$$

\noindent In the case of  a covariant $2$-tensor $\phi$ on $\Omega$, $\| \phi\|_{B^s_{p,q}(\Omega)}$ will be defined similarly in terms of its Cartesian components.
\end{definition}

\begin{remark*}
The Besov norm $\| \phi\|_{B^s_{p,q}(\Omega)}$ of a given tensor $\phi$ on $\Omega \subset \mathbb R^3$ as defined above is not affected by rotations and translations of the Cartesian coordinate system used to express the components of $\phi$.
\end{remark*}

\section{Proof of main theorem (Theorem~\ref{thm:main}) \label{sec:Proof}}

For the remainder of this section, we will assume that $\Omega$ and $f: \Omega \rightarrow \mathbb R$ have been fixed as in the statement of Theorem~\ref{thm:main}, and similarly for $g,k$ and $g_0, k_0$. We will also assume, for the sake of contradiction, that $\Omega$ contains a smooth, closed, embedded trapped surface $S$; without loss of generality, we will assume that $S$ is connected.

\begin{definition}\label{def:Gauss map}
We will define the null Gauss map $\Phi: S \rightarrow \mathbb{S}^2$ so that for any point $p\in S$, $\Phi(p)$ is the point on $\mathbb S^2$ in the direction of the Minkowskian null normal $L^{\overline{f}(S)}[p]$, i.e.
\[
\Phi(p) = \omega \Leftrightarrow L^{\overline{f}(S)}[p] = L_\omega
\]
(see Definition \ref{def:Normalized null normal} for the definition of the normalized  null normal $L^{\overline{f}(S)}$ and Definition \ref{def:Null vector sphere} for the definition of $L_\omega$).
\end{definition}

\begin{proposition}\label{prop:S+}
The set $S^+ \subseteq S$ (see Definition \ref{def:Null convex hull}) is non-empty and satisfies 
\begin{equation} \label{eq:Lower bound trchi}
\int_{S^+} (\tr \chi_0)^2 \, \mathrm{dVol}_{S,g_0} \ge 16\pi.
\end{equation}

\end{proposition}


Before we begin the proof of Proposition \ref{prop:S+}, we need to establish the following result regarding the pullback of certain integrals via the null Gauss map $\Phi$:

\begin{lemma}\label{lem:Inverse Gauss map}
Let $U \subset S$ be an open subset such that the null Gauss map $\Phi: U \to \Phi(U)$ is a local diffeomorphism (see Definition \ref{def:Gauss map} for the definition of $\Phi$). Suppose $r:U\to \mathbb R$ is a smooth function. Then 
\begin{equation}\label{eq:change.of.variables}
\int_U \Big(r\circ \Phi\Big) (\mathrm{tr} \chi_0 )^2 \, \mathrm{dVol}_{S,g_0} \geq 4 \int_{\Phi(U)} r \, \mathrm{dVol}_{\mathbb S^2}.
\end{equation}
\end{lemma}
\begin{proof}
We will perform the computations in local coordinates. We will assume, without loss of generality (by considering a smaller coordinate patch, if necessary), that $f:\Omg\to \mathbb R$ is a smooth function satisfying $\inf_\Omg (1-|\rd f|_e) > 0$ and $S$ is locally given by $$x^3 = \psi(x^1,x^2)$$ for some smooth function $\psi$ satisfying $(\partial_1 \psi)^2+(\partial_2\psi)^2 < \f14\inf_\Omg (1-|\rd f|^2_e)$.

\pfstep{Step~1: Computations} We will introduce the notation $\bnab$ to denote the gradient of a function in the $(x^1,x^2)$ variables, so that $\bnab \psi = (\rd_1 \psi, \rd_2 \psi)$, $|\bnab \psi|^2 = (\rd_1 \psi)^2 + (\rd_2 \psi)^2$, and similarly for $\bnab f$ and $|\bnab f|$. Moreover, for the remainder of the proof, we will adopt the convention that capital Latin indices run through $A,B=1,2$.

We now collect some computations in local coordinates.
\begin{enumerate}
    \item Locally, the tangent space of $\overline{f}(S)$ is spanned by 
    \begin{equation}\label{eq:(S)tangent}
        \srd_A\doteq \rd_A + (\rd_A \psi) \rd_3 + (\rd_A f+ \rd_3 f \, \rd_A \psi)\rd_t,\quad A=1,2.
    \end{equation}
    \item The Minkowskian null normal $L^{\overline{f}(S)} = (1,v^1,v^2,v^3)$ can be computed as follows: The condition \eqref{Normal fS} for the tangent vectors in \eqref{eq:(S)tangent} implies that $-\rd_A f -\rd_3 f \cdot \rd_A \psi+ \de_{AB} v^B + v^3 \rd_A \psi = 0$ (where $\de_{AB}$ denotes the Kronecker delta), or equivalently,
    \begin{equation}\label{eq:vA}
        \de_{AB} v^B = \rd_A f + \rd_3 f \cdot \rd_A \psi- v^3 \rd_A \psi,\quad A=1,2.
    \end{equation}
    Using \eqref{eq:vA} together with the condition that $|v|_e = 1$, we obtain 
    \begin{equation}
        | \bnab f+\rd_3 f \bnab \psi|^2 -2v^3 \bnab \psi \cdot (\bnab f+\rd_3 f \bnab \psi) + (v^3)^2 |\bnab \psi|^2 + (v^3)^2 = 1,
    \end{equation}
    which yields the following expression for $v^3$:
    \begin{equation}\label{eq:v3}
        v^3 = \f{\bnab \psi \cdot \bnab f + \rd_3 f |\bnab \psi|^2 \pm \sqrt{(\bnab \psi \cdot \bnab f + \rd_3 f |\bnab \psi|^2)^2 + (1+|\bnab \psi|^2)(1-|\bnab f+\rd_3 f \bnab \psi|^2)}}{1+|\bnab \psi|^2}.
    \end{equation}
    Since $1-|\bnab f+\rd_3 f \bnab \psi|^2=1-|\rd f|^2_e + (1-|\bnab\psi|^2)(\rd_3 f)^2 - 2 \rd_3 f \, \bnab \psi \cdot \bnab f >0$ (in view of our assumption that $|\bnab \psi|^2 < \f14 (1-|\rd f|_e^2)$), the expression \eqref{eq:v3} implies that $v^3$ does not change sign in the region covered by the local coordinate chart. Without loss of generality, we will thus assume that $v^3>0$ (which corresponds to considering the $+$ solution in \eqref{eq:v3}).
    \item In these coordinates, the null second fundamental form of $\overline{f}(S)$ associated to $L^{\overline{f}(S)}$ can be computed as follows (using the relation $v^3 = \sqrt{1-|\bar{v}|^2}$):
    \begin{equation}\label{eq:DL.local}
        \begin{split}
            DL^{\overline{f}(S)}(\srd_A,\srd_B) = &\: m\Big( \rd_A v^C \rd_C - \f{\de_{CD}v^D(\rd_A v^C)}{\sqrt{1-|\bar{v}|^2}}\rd_3, \rd_B + (\rd_B \psi) \rd_3 + (\rd_B f + \rd_3 f \, \rd_B \psi) \rd_t \Big) \\
            = &\: \rd_A v^C \Big(\de_{BC} - \f{\de_{CD} v^D(\rd_B \psi)}{\sqrt{1-|\bar{v}|^2}} \Big).
        \end{split}
    \end{equation}
    Define the null shape operator $(DL^{\overline{f}(S)})^{\sharp_{\sg_0}}$ to be the $(1,1)$-tensor  given in local coordinates by 
    \begin{equation}\label{eq:W.map}
    \begin{split}
        \big((DL^{\overline{f}(S)})^{\sharp_{\sg_0}}\big)_A^E \doteq &\: (\sg_0^{-1})^{BE} (DL^{\overline{f}(S)})_{AB} 
        =  (\sg_0^{-1})^{BE} (\rd_A v^C) \Big(\de_{BC} - \f{\de_{CD} v^D(\rd_B \psi)}{\sqrt{1-|\bar{v}|^2}} \Big),
    \end{split}
    \end{equation}
    where $\sg_0$ is the metric induced on $S$ by $g_0$ and in the last line we have used \eqref{eq:DL.local}.
    \item We will now compute the determinant of the matrix appearing as the third factor on the right-hand side of \eqref{eq:W.map}. Using \eqref{eq:vA} and \eqref{eq:v3}, we have
    \begin{equation}\label{eq:the.det}
        \begin{split}
            &\: \det \Big(\de_{BC} - \f{\de_{CD} v^D(\rd_B \psi)}{\sqrt{1-|\bar{v}|^2}} \Big) \\
    = &\: \Big(1 - \f{v^1 \rd_1 \psi}{\sqrt{1-|\bar{v}|^2}}\Big)\Big(1 - \f{v^2 \rd_2 \psi}{\sqrt{1-|\bar{v}|^2}}\Big) - \f{v^1 v^2 (\rd_1 \psi) (\rd_2 \psi)}{1 - |\bar{v}|^2} \\
    = &\: 1 - \f{v^A \rd_A \psi}{\sqrt{1-|\bar{v}|^2}}  = 1 - \f{(\bnab f)\cdot (\bnab \psi) +(\rd_3 f- v^3) |\bnab \psi|^2 }{\sqrt{1 - |\bar{v}|^2}} = \f{v^3(1+|\bnab \psi|^2) - \rd_3 f |\bnab \psi|^2 - (\bnab f)\cdot (\bnab \psi)}{v^3} \\
            = &\: \f{\sqrt{(\bnab \psi \cdot \bnab f + \rd_3 f |\bnab \psi|^2)^2 + (1+|\bnab \psi|^2)(1-|\bnab f+\rd_3 f \bnab \psi|^2)}}{v^3}.
        \end{split}
    \end{equation}
    \item We compute that 
    \begin{equation}
        (\sg_0)_{AB} = \de_{AB} + (\rd_A \psi)(\rd_B \psi) - (\rd_A f+ \rd_3 f \, \rd_A \psi)(\rd_B f+ \rd_3 f \, \rd_B \psi),
    \end{equation}
    and thus
    \begin{equation}\label{eq:detgS}
        \det \sg_0 = 1+ |\bnab \psi|^2 - |\bnab f + \rd_3 f \bnab \psi|^2 - (\rd_1 \psi \rd_2 f - \rd_2 \psi \rd_1 f)^2.
    \end{equation}
    In particular, the following lower bound holds:
    \begin{align}\label{eq:gS.lower.bound}
    \det \sg_0 & = 1 + (1-|\partial f|_e^2)|\bnab \psi|^2 - |\bnab f|^2 -2\rd_3 f (\bnab f \cdot \bnab \psi)+ (\bnab f \cdot \bnab \psi)^2 
    \\
    & \ge (1 + |\bnab \psi|^2) (1 - |\partial f|^2_e). \nonumber 
	\end{align}
    \item We finally collect some computations concerning volume forms. First, the volume forms $\mathrm{dVol}_{S,g_0}$ and $\mathrm{dVol}_{\mathbb S^2}$ are given in our coordinates as follows:
    \begin{align}
        \mathrm{dVol}_{S,g_0} =  \sqrt{\det \sg_0}\, \ud x^1 \, \ud x^2, \label{eq:dVol}\quad \mathrm{dVol}_{\mathbb S^2} =\sqrt{1+ \f{|\bar{v}|^2}{1-|\bar{v}|^2}}\, \ud v^1\, \ud v^2  = \f 1{v^3} \, \ud v^1\, \ud v^2.
    \end{align}
    Moreover, the pull-back of the volume form is given by
    \begin{equation}\label{eq:dVol.pullback}
        |\det (Dv)|\, \ud x^1 \, \ud x^2 = \Phi^*(\ud v^1\, \ud v^2),
    \end{equation}
    where $Dv$ denotes the matrix whose entries are given by $\rd_A v^B$.
\end{enumerate} 

\pfstep{Step~3: Proving the identity}
Starting with \eqref{eq:W.map}, and using \eqref{eq:the.det} and \eqref{eq:detgS}, we obtain
\begin{equation}\label{eq:Dv.est}
\begin{split}
    |\det(Dv)| = &\: \Big(\det (DL^{\overline{f}(S)})^{\sharp_{\sg_0}} \Big) \Big(\det \sg_0 \Big) \Bigg(\det\Big(\de_{BC} - \f{\de_{CD} v^D(\rd_B \psi)}{\sqrt{1-|\bar{v}|^2}} \Big)\Bigg)^{-1} \\
    = &\: \f{v^3 \sqrt{1+ |\bnab \psi|^2 - |\bnab f + \rd_3 f \bnab \psi|^2 - (\rd_1 \psi \rd_2 f - \rd_2 \psi \rd_1 f)^2}}{\sqrt{(\bnab \psi \cdot \bnab f + \rd_3 f |\bnab \psi|^2)^2 + (1+|\bnab \psi|^2)(1-|\bnab f+\rd_3 f \bnab \psi|^2)}}  \Big(\det \sg_0 \Big)^{\f 12}\Big(\det (DL^{\overline{f}(S)})^{\sharp_{\sg_0}} \Big) \\
    =&\: v^3 \Big(\det \slashed{g}_{0} \Big)^{\frac 12}\Big(\det (DL^{\overline{f}(S)})^{\sharp _{\slashed{g}_{0}}} \Big).
\end{split}
\end{equation}

By a standard partition of unity argument, we can assume that $U$ lies inside a local coordinate patch that we are considering and $\Phi$ is one-to-one when restricted to $U$. Using the volume forms computations in \eqref{eq:dVol}, the transformation in \eqref{eq:dVol.pullback}, and the formula in \eqref{eq:Dv.est}, we obtain
\begin{equation}\label{eq:cov.almost.the.end}
    \begin{split}
        &\: \int_{\Phi(U)} r \, \mathrm{dVol}_{\mathbb S^2} = \int_{\Phi(U)} r\, \f 1{v^3} \, \ud v^1\, \ud v^2 = \int_U (r\circ \Phi) |\det(Dv)| \f 1{v^3} \, \ud x^1 \, \ud x^2 \\
        = &\: \int_U (r\circ \Phi)  \Big(\det (DL^{\overline{f}(S)})^{\sharp_{\sg_0}} \Big) \Big(\det \sg_0 \Big)^{\f 12} \, \ud x^1 \, \ud x^2 \\
        = &\: \int_U (r\circ \Phi)  \Big(\det (DL^{\overline{f}(S)})^{\sharp_{\sg_0}} \Big) \, \mathrm{dVol}_{S,g_0}.
    \end{split}
\end{equation}

Finally, it is easy to get from \eqref{eq:cov.almost.the.end} to the desired estimate \eqref{eq:change.of.variables} after noting that 
\begin{itemize}
    \item $\det (DL^{\overline{f}(S)})^{\sharp_{\sg_0}} \leq \f 14 \Big(\big((DL^{\overline{f}(S)})^{\sharp_{g_s}}\big)_A^A\Big)^2$ by the AM-GM inequality, and 
    \item $\big((DL^{\overline{f}(S)})^{\sharp_{g_s}}\big)_A^A = \mathrm{tr}\chi_0$. \qedhere
\end{itemize}
\end{proof}

\noindent \emph{Proof of Proposition \ref{prop:S+}.} The restriction of $\Phi$ to $S^+$ maps \textbf{onto} $\mathbb S^2$ (thus, in particular, $S^+$ is non-empty). This can be deduced as follows: For any $\omega \in \mathbb S^2$, define
\[
u_\omega = \max \big\{ u: \, \overline{f}(S) \subset W_{\omega, u}\big\}
\]
($u_\omega$ is well-defined and satisfies $-\infty < u_\omega < +\infty$ since $\overline{f}(S)$ is non-empty and compact). Then, 
\[
\overline{f}(S) \cap \Pi_{\omega, u_\omega} \neq \emptyset,
\]
 since otherwise $\overline{f}(S)$ would be contained in the interior of $W_{\omega, u_\omega}$; this would imply that $\overline{f}(S) \subset W_{\omega, u_\omega+\delta}$ for some $\delta>0$ (since $\overline{f}(S)$ is compact) , contradicting the definition of $u_\omega$. Let $p_\omega \in S$ be chosen such that the image $q_\omega$  of $p_\omega$ in  $\overline{f}(S)$ lies in $\overline{f}(S) \cap \Pi_{\omega, u_\omega}$. Then $q_\omega$ necessarily lies on $\partial \mathbb K_+[S]$ since, by definition of $\mathbb K_+[S]$, we have
\[
\mathbb K_+[S] \subset W_{\omega,u_\omega}.
\]
Thus,  $p_\omega \in S^+$. Moreover,
\[
L^{\overline{f}(S)}[p_\omega] = L_\omega
\]
since $(1,\omega)$ is the null generator of $\Pi_{\omega,u_\omega}$ and  $\Pi_{\omega,u_\omega}$  contains the tangent space of $\overline{f}(S)$ at $\overline{f}(p_\omega)$ (since $\overline{f}(p_\omega) \in \Pi_{\omega,u_\omega}$ and $\overline{f}(S) \subset W_{\omega, u_\omega}$).

Let $\mathcal N \subset S$ be the set of points where $D\Phi: TS \rightarrow T\mathbb S^2$ is degenerate. By Sard's lemma, we know that $\Phi(\mathcal{N})$ is of zero measure in $\mathbb{S}^2$ (with respect to $\mathrm{dVol}_{\mathbb S^2}$). Moreover, $\mathcal N$ is a compact subset of $S$ (since $S$ is compact and $\mathcal N$ is necessarily closed), thus $\mathbb S^2 \setminus \Phi(\mathcal N)$ is an open subset of $\mathbb S^2$ of full measure. Let $\mathcal{V}$ be any open subset of $S \setminus \mathcal{N}$ containing $S^+\setminus \mathcal{N}$. Since $\Phi$ maps $S^+$ onto $\mathbb S^2$, we deduce that $\Phi(\mathcal{V})= \mathbb S^2 \setminus \Phi(\mathcal{N})$. Since $D\Phi$ is invertible on $\mathcal V \subset S\setminus \mathcal{N}$, the map $\Phi: \mathcal{V} \rightarrow \mathbb S^2 \setminus \Phi(\mathcal{N})$ is a covering map. Applying Lemma \ref{lem:Inverse Gauss map} for $r=1$ and $\mathcal{U} = \mathcal{V}$, we therefore deduce that 
\begin{align*}
\int_{\mathcal{V}} (\tr \chi_0)^2\,\mathrm{dVol}_{S,g_0} & \ge 4 \int_{\mathbb S^2 \setminus \Phi (\mathcal N)}\,\mathrm{dVol}_{\mathbb S^2} = 16\pi.
\end{align*}
Since the above bound holds for any open set $\mathcal{V}$ containing $S^+ \setminus \mathcal N$, we readily infer \eqref{eq:Lower bound trchi}. \qed

The following result states that $S^+$ can be separated into a fixed number of pieces (depending only on $\inf_\Omega (1-|\rd f|_e)$) which can be represented as graphs of smooth functions over planes in $\mathbb R^3$; the way $S^+$ was defined is crucial for the validity of this statement.  

\begin{lemma}\label{lem:patches}
Setting
\[
N \doteq \Big\lceil \sup_\Omega \Big( \frac{100}{1-|\rd f|_e} \Big)^2 \Big\rceil,
\]
there exist relatively open sets $\{ U_i\}_{i=1}^N \subseteq S$ with the following properties:
\begin{itemize}
\item The union of $\{U_i\}_{i=1}^N$ covers $S^+$.
\item For any $i\in \{1, \ldots, N\}$, the subset $S \cap U_i$ of $S$ can be written as a graph in the following way: There exists a rotation $\mathbb U_i \in SO(3)$ such that, in the coordinate system $(y^1, y^2, y^3)$ obtained from $(x^1, x^2, x^3)$ after rotating by $\mathbb U_i$, the subset $U_i$ of $S$  can be expressed  as a graph
$$ U_i \subseteq \{(y_1, y_2, y_3): y^3 = \psi_i(y^1, y^2)\},$$
where $\psi_i:V_i \to \mathbb R$ is a smooth function defined on an open set $V_i \subseteq \mathbb R^2$.
\item For any $i\in \{1, \ldots, N\}$, the function $\psi_i$ satisfies the gradient bound
\begin{equation}\label{eq:Gradient bound}
\sqrt{|\rd_1 \psi_i|^2+|\rd_2\psi_i|^2} \leq \sup_\Omega \Big(\f{5}{\sqrt{1-|\rd f|_e}}\Big).
\end{equation}
\item For any $i\in \{1, \ldots, N\}$, the components of the null normal $L^{\overline{f}(S)} = (1, v^1, v^2, v^3)$ on $U_i$ expressed with respect to the rotated coordinate system $(y^1,y^2,y^3)$ on $\mathbb R^3$ as above satisfy
\begin{equation}\label{eq:Smallness angular variation}
\sqrt{|v^1|^2 + |v^2|^2} \le \frac{12}{\sqrt{N}}.
\end{equation}
\end{itemize}
\end{lemma}

\begin{remark*}
In the simpler case when $f=0$, $S^+$ is contained in the boundary of the convex hull of $S$. In that case, Lemma \ref{lem:patches} follows readily for any $N \ge 6$ from the fact that the boundary of any convex body in $\mathbb R^3$ can be split into $6$ pieces, each of which representable as the graph of a function over a coordinate plane.
\end{remark*}

\begin{proof}
Pick a collection of points $\{w_i\}_{i=1}^N$ in $\mathbb S^2$ with the property that the union of the balls 
\begin{equation}\label{eq:Theta i}
\Theta_i \doteq \big\{v \in \mathbb S^2: |v-w_i|_e < \frac{10}{\sqrt{N}} \big\}
\end{equation}
covers the whole of $\mathbb{S}^2$. For any $i\in\{1, \ldots, N\}$, let us define the sets $\tilde{U}_i \subset S^+$ as follows:
\[
\tilde{U}_i \doteq \Phi^{-1}(\Theta_i) \cap S^+,
\]
where $\Phi: S \rightarrow \mathbb S^2$ is the null Gauss map introduced in Definition \ref{def:Gauss map}. This implies, in particular, that for any $p\in \tilde{U}_i$, the Minkowskian outgoing null normal $L^{\overline{f}(S)}[p]$ is of the form $(1, \omega)$ for some $\omega \in \Theta_i$. 

For any $i\in \{1, \ldots, N \}$, fix a rotation $\mathbb U_i \in SO(3)$ so that, in the coordinate system $(y^1, y^2, y^3)$ on $\mathbb R^3$ obtained after rotating the fixed Cartesian system by $\mathbb U_i$, the vector field $\partial_{y^3}$ is equal to $w_i$. We will now show that there exists an open neighborhood $U_i$ of $\tilde{U}_i$ inside the surface $S$ such that $U_i$ is the graph of a smooth function $y^3=\psi_i(y^1,y^2)$. This claim will follow immediately by showing the following (and using the implicit function theorem):
\begin{enumerate}
    \item For any $p\in \tilde{U}_i$, the tangent space $T_p S$ is transversal to $w_i=(0,0,1)$.
    \item Any coordinate line of the form $y^1, y^2 =\textrm{const}$ in $\mathbb R^3$ intersects $\tilde{U}_i$ in \textbf{at most one} point.
\end{enumerate}
Establishing the property (1) is going to be a direct consequence of the definition of the set $\tilde{U}_i$: For any $p\in \tilde{U}_i$, the null normal $L^{\overline{f}(S)}[p]$ is of the form
\[
L^{\overline{f}(S)}[p] = (1,\omega) \quad \text{for some } \omega\in \Theta_i \subset \mathbb S^2, 
\]
while any vector $X=(X^1, X^2, X^3)\in T_p S$ satisfies 
\[
\langle \overline{f}^* X, L^{\overline{f}(S)}[p] \rangle_m = 0,
\]
where $\overline{f}^*X = (D_X f,X^1, X^2, X^3) \in T_{f(p)} \overline{f}(S)$ is simply the push-forward of $X$ along the map $f$. Therefore, in order to show that $w_i$ is transversal to $T_p S$, it suffices to show that $\langle \overline{f}^* w_i, L^{\overline{f}(S)}[p] \rangle_m \neq 0$. This follows readily by computing:
\[
 \langle \overline{f}^* w_i, L^{\overline{f}(S)}[p] \rangle_m = - (w_i)^j \partial_j f + \langle \omega, w_i \rangle_e > 0
\]
since, for any $\omega \in \Theta_i$: 
\begin{align*}
-|\rd f|_e + \langle \omega, w_i \rangle_e & \ge -|\rd f|_e + 1 - |\omega-w_i|_e \\
& \ge -|\rd f|_e + 1 - \frac{10}{\sqrt{N}} \\
& \ge -|\rd f|_e + 1 - \frac{1-|\rd f|_e}{10}\\
& \ge \frac9{10} (1-|\rd f|_e) >0.
\end{align*}

Establishing the property (2) is a bit trickier, and here we are going to make use of the way the set $S^+$ was defined. Assume, for the sake of contradiction, that there exists a point  $(\bar{y}^1, \bar{y}^2)\in \mathbb R^2$ such that the straight line
\[
l = \{ (y^1, y^2, y^3) = ( \bar{y}^1, \bar{y}^2, \tau), \tau \in \mathbb R  \}
\]
intersects $\tilde{U}_i$ at two points $p_1, p_2$. By switching the roles of $p_1$ and $p_2$ if necessary, we will assume that $$p_1 = ( \bar{y}^1, \bar{y}^2, \tau_1) \text{ and } p_2= ( \bar{y}^1, \bar{y}^2, \tau_2) \text{ with } \tau_1 < \tau_2.$$ Let $\omega_1, \omega_2 \in \Theta_i \subset \mathbb S^2$ be the directions of the corresponding null normals at $p_1, p_2$, i.e.~$L^{\overline{f}(S)}[p_1]=(1,\omega_1)$ and $L^{\overline{f}(S)}[p_2]=(1,\omega_2)$. Since $\tilde{U}_i \subset S^+$, the definition of $S^+$ (see Definition \ref{def:Null convex hull} and the comments above Lemma \ref{lem:Positivity of DL}) implies that the surface $\overline{f}(S)$ is contained in the Minkowskian half-spaces 
\[
W_j = \big\{ (t,y): \, t - \langle y , \omega_j \rangle_e \ge f(p_j) - \langle p_j , \omega_j \rangle_e \big\} \quad j=1,2,
\]
which are the future half spaces determined by the null hyperplanes $\Pi^{\overline{f}(S)}[p_j]$. In particular, $$\big(f(\bar{y}^1, \bar{y}^2, \tau_2), \bar{y}^1, \bar{y}^2, \tau_2 \big) \in W_1$$ and therefore
\[
f(\bar{y}^1, \bar{y}^2, \tau_2) - \tau_2 \omega^3_1 \ge f(\bar{y}^1, \bar{y}^2, \tau_1) - \tau_1 \omega^3_1
\]
or, equivalently, since $\tau_1 < \tau_2$
\[
\frac{f(\bar{y}^1, \bar{y}^2, \tau_2) - f(\bar{y}^1, \bar{y}^2, \tau_1)}{\tau_2 - \tau_1} \ge \omega^3_1 
\]
which is a contradiction, since 
\begin{align*}
\omega_1^3 - \sup_\Omega|\rd f|_e & = \langle \omega_1 , w_i \rangle_e  - \sup_\Omega|\rd f|_e \\
 & \ge 1 - \frac{10}{\sqrt{N}} - \sup_\Omega|\rd f|_e \\
 & \ge \frac{9}{10} \big( 1 - \sup_\Omega |\rd f|_e \big)> 0.
\end{align*}
Thus, having established both Properties (1) and (2) above, we infer that an open neighborhood $U_i$ of $\tilde{U}_i$ in $S$ can be written as the graph of a smooth function $y^3 = \psi_i (y^1, y^2)$.

In the $(t, y^1, y^2, y^3)$ coordinate system (where, as before, $(y^1, y^3, y^3)$ is the rotated Cartesian system in which $w_i = (0,0,1)$), the components of the vector $L^{\overline{f}(S)}=(1, v^1, v^2, v^3)$ satisfy the relation
\[
v^3 \partial_A \psi = \partial_A f + \rd_3 f \rd_A \psi- \delta_{AB} v^{B}, \quad A = 1,2.
\]
Thus, noting that $|\bnab f| \le \sqrt{1-(\rd_3 f)^2}$, $|\rd f|_e <1$ and, in $\tilde{U}_i=\Phi^{-1}(\Theta_i)$, we have (by the definition \eqref{eq:Theta i} of $\Theta_i$) $|\bar v| \le \f{10}{\sqrt{N}}$ and $(1-v^3)^2+ 1-(v^3)^2<\f{100}{N}$ and, thus, $v^3 > 1-\frac{50}{N}$, we infer that 
\[
|\bnab \psi_i|_{L^\infty(\tilde{U}_i)} \le  \big| \f{\bnab f}{v^3-\rd_3 f} \big|_{L^\infty(\tilde{U}_i)} + \big| \f{\bar v}{v^3-\rd_3 f} \big|_{L^\infty(\tilde{U}_i)} \le 4 \sup_\Omega (1-|\rd f|_e)^{-\f12}.
\]
 Thus, by possibly considering a slightly smaller open neighborhood $U_i$ of $\tilde{U}_i$, and using the smoothness of $S$, we infer that \eqref{eq:Gradient bound} holds on $U_i$. 

The bound \eqref{eq:Smallness angular variation} is a direct consequence of the definition \eqref{eq:Theta i} of $\Theta_i$ after choosing $U_i$ smaller if necessary.
\end{proof}

The following lemma is obtained by comparing the geometries of $\overline{f}(S)$ with respect to $(g_0, k_0)$ and $(g,k)$ and using the assumption that $S$ is a trapped surface for $(g,k)$.

\begin{lemma}\label{lem:Comparison geometries}
Let $S$, $S^+$ and $\{ U_i \}_{i=1}^N$, $\{ \psi_i \}_{i=1}^N$ be as in the statement of Lemma \ref{lem:patches} and let $(g,k)$ be a smooth pair of tensors on $\Omega$ as in the statement of Theorem \ref{thm:main}, such that $S$ is a trapped surface for $(\Omega,g,k)$. Assume also that the parameter $\ep_0$ in  Theorem \ref{thm:main} satisfies $\ep_0 \leq \f 1{100}\inf_\Omg (1-|\rd f|_e)$. Then, in each $U_i$, the Minkowskian null expansion $\tr\chi_0$ satisfies the pointwise bound (in the rotated Cartesian system $(y^1, y^2, y^3)$ associated to $U_i$):
\begin{equation}\label{eq:Trace.trapped}
\big| \tr\chi_0|_{U_i\cap S^+}\big| \le \frac{C}{\inf_\Omega(1-|\rd f|_e)^{7}}  \big(|\partial(g-g_0)|_e + |k-k_0|_e + \epsilon_0 (|\partial^2 \psi_i|_e +|\partial^2 f|_e+1) \big),   
\end{equation}
where $C>0$ is an absolute constant.
\end{lemma}

\begin{proof}
Using capital letters for indices associated to the $(y^1, y^2)$ chart on $U_i$ and small letters for indices associated to $(y^1, y^2, y^3)$, we can explicitly compute the induced second fundamental form $h_g$ on $U_i \cap S$ (which can be viewed as the graph $y^3=\psi_i (y^1, y^2)$), by noting that $(h_g)_{AB} = \nabla^2_{a b } F T^a_A T_B^b$, where\footnote{Here, we are implicitly making use of the fact that the exterior pointing normal to $S$ with respect to $g$ satisfies the sign condition $\langle N_g, \partial_3 \rangle_g >0$, which is the only sign choice consistent with the fact that, for the metric $g_0$, we have $\langle N_{g_0}, \partial_3 \rangle_{g_0} >0$ (which in turn follows from our assumption that $v^3=+\sqrt{1-|\bar{v}|^2}$ and $|\bar{v}| \le \frac{1}{5}\inf_{\Omega}(1-|\rd f|_e)$ for $L^{\overline{f}(S)}=(1,v^1,v^2,v^3)$ on $U_i$).} 
\begin{align*}
F(y)=\lambda(y)\cdot(y^3-\psi_i(y^1, y^2)), \\
\lambda|_S = \frac{1}{\sqrt{g^{AB}\partial_A \psi_i \partial_B \psi_i - 2g^{3A}\partial_A \psi_i+g^{33}}}
\end{align*}
and
\[
T_A^a = \begin{cases} 1, \quad \text{if }a=A, \\ 0 \quad \text{if } a\in \{1,2\} \text{ and } a\neq A,\\ \partial_A \psi_i \quad \text{if } a=3. \end{cases}
\]
In particular, we obtain the expression
\begin{align*}
(h_g)_{AB} = \lambda(y) \cdot \Big(  - & \partial^2_{AB} \psi_i + \Gamma^C_{AB} \partial_C \psi_i + \Gamma^C_{3A}\partial_C \psi_i \partial_B \psi_i + \Gamma^C_{3B}\partial_C \psi_i \partial_A \psi_i \\
& -\Gamma^3_{AB}-\Gamma^3_{A3}\partial_B \psi_i -\Gamma^3_{B3}\partial_A \psi_i - \Gamma^3_{33}\Big)
\end{align*}
where $\Gamma_{ab}^c$ are the Christoffel symbols of $g$ in the rotated $(y^1, y^2, y^3)$ chart.

Therefore, we can estimate using the gradient bound \eqref{eq:Gradient bound} for $\psi_i$, the trivial Sobolev estimate $$|g-g_0|_e \lesssim \| g - g_0\|_{B^{\f32}_{2,1}}\lesssim \epsilon_0$$ and the bound $|g^{-1}-g_0^{-1}|_e \lesssim \big(\inf_\Omega(1-|\rd f|_e)\big)^{-2} \epsilon_0$ (following from the explicit expression \eqref{eq:Flat metric} for $g_0$ and the assumption $\ep_0 \leq \f 1{100}\inf_\Omg (1-|\rd f|_e)$):
\begin{equation}\label{eq:h.difference}
    |h_g - h_{g_0} |_e \lesssim \frac1{\inf_\Omega(1-|\rd f|_e)^4}  \Big(|\partial(g-g_0)|_e + \epsilon_0 (|\partial^2 \psi_i|_e+|\partial^2 f|_e +1) \Big).
\end{equation}

As a result, using also \eqref{eq:trch0}, we can bound:
\begin{align}
\big|\tr\chi_0 - & \zeta_S\tr_{\sg}(k+h_g) \big|  = \zeta_S \cdot \Big(\big|\tr_{\sg_0}(k_0+h_{g_0}) - \tr_{\sg}(k+h_g) \big|\Big)  \nonumber \\
& \lesssim \f1{\inf_\Omega(1-|\rd f|_e)}\Big(|g^{-1}-g_0^{-1}|_e\big(|k|_e+|h_{g}|_e\big) + |g_0^{-1}|_e\big(|k-k_0|_e+|h_g-h_{g_0}|_e\big) \Big)  \nonumber\\
&  \lesssim \frac1{\inf_\Omega(1-|\rd f|_e)^{6}}  \Big(|\partial(g-g_0)|_e + |k-k_0|_e + \epsilon_0 \big( |\partial^2 \psi_i|_e +|\partial g_0|_e+|k_0|_e+1\big) \Big) \nonumber \\
&\lesssim \frac1{\inf_\Omega(1-|\rd f|_e)^{7}}  \Big(|\partial(g-g_0)|_e  + |k-k_0|_e + \epsilon_0 \big( |\partial^2 \psi_i|_e +|\partial^2 f|_e+1\big) \Big), \label{eq:Upper bound geometry diff}
\end{align}
where, in the last line above, we made use of the explicit form \eqref{eq:Flat metric}--\eqref{eq:Flat k} of $(g_0,k_0)$ in terms of $f$.

In view of the fact that $\tr\chi_0 \ge 0$ on $S^+$ (by Lemma~\ref{lem:Positivity of DL}) and the assumption that $S$ is trapped for $(g,k)$ and thus $\tr_{\sg}(k+h_g)<0$ (using also $\zeta_S>0$ in \eqref{eq:trch0}), we infer from \eqref{eq:Upper bound geometry diff} that, on $S^+\cap U_i$:
\[
|\tr\chi_0| \lesssim \big(\tr\chi_0 -\zeta_S \tr_{\sg}(k+h_g) \big) \lesssim \frac1{\inf_\Omega(1-|\rd f|_e)^{7}}  \Big(|\partial(g-g_0)|_e + |k-k_0|_e + \epsilon_0 \big( |\partial^2 \psi_i|_e +|\partial^2 f|_e+1\big) \Big),
\]
which establishes \eqref{eq:Trace.trapped}.
\end{proof}

\begin{lemma}\label{lem:Upper bound hessian}
Let $U_i \subset S$, $(y^1, y^2, y^3)$ and $\psi_i : V_i \subset \mathbb R^2 \rightarrow \mathbb R$  be as in the statement of Lemma \ref{lem:patches}. Then the following pointwise bound holds for $\psi_i$:
\begin{equation}\label{eq:Upper bound hessian}
|\partial^2 \psi_i  | \le \frac{C}{\inf_{\Omega}(1-|\rd f|_e)^3} \cdot \big( \tr \chi_0 + |\partial^2 f|_e\big) \quad \text{on }   U_i \cap S^+,  
\end{equation}
where  $C >0$ is an absolute constant.
\end{lemma}

\begin{proof}
In this proof, we will only work with the coordinate chart $(y^1, y^2)$ induced on $U_i$ by $\psi_i$ (recall that $U_i = \{ y^3 = \psi_i (y^1, y^2) \}$); we will use capital letters to denote indices associated to this chart. The components of the vector field $L^{\overline{f}(S)} = (1, v^1, v^2, v^3)$ can be computed as in \eqref{eq:vA}, i.e.
\[
\delta_{AB} v^B = \partial_A f + \rd_3 f \rd_A\psi_i - \sqrt{1-|\bar{v}|^2} \rd_A \psi_i, \quad A = 1,2,
\]
where $\bar{v} = (v^1, v^2)$. Solving the above relation with respect to $\partial_A \psi_i$ (recalling that $\sqrt{1-|\bar v|^2} -\rd_3 f> \f12 \inf_\Omega (1-|\rd f|_e)$ on $U_i$ as a consequence of \eqref{eq:Smallness angular variation}) and differentiating once more, we obtain 
\[
\partial^2_{AB} \psi_i  =  \frac{1}{\sqrt{1-|\bar{v}|^2}-\rd_3 f}\Big( \partial^2_{AB} f - \big( \delta_{CA} +  \f{\delta_{AJ} v^J-\rd_A f}{\sqrt{1-|\bar{v}|^2}-\rd_3 f} \delta_{IC}\frac{v^I}{\sqrt{1-|\bar{v}|^2}} \big)\partial_B v^C + \f{\delta_{AC} v^C-\rd_A f}{\sqrt{1-|\bar{v}|^2}-\rd_3 f} \rd^2_{3B} f\Big).  
\]
Using the bound \eqref{eq:Smallness angular variation} for $\bar{v}$ on $U_i$, we thus obtain
\begin{equation}\label{Hessian psi}
|\partial^2 \psi_i|  \le \f{4}{\inf_\Omega (1-|\rd f|_e)}\Big( |\partial^2 f|_{e} + 10| \bnab \bar{v}| \Big).  
\end{equation}

The Minkowskian null shape operator $\big((DL^{\overline{f}(S)})^{\sharp_{\sg_0}}\big)_A^B = (\sg_0^{-1})^{BC} (DL^{\overline{f}(S)})_{AC}$ can be computed as in \eqref{eq:W.map}:
\begin{equation}\label{eq:Null shape again}
\big( (DL^{\overline{f}(S)})^{\sharp_{\sg_0}}\big)_A^B = \mathbb M_C^B \cdot \partial_A v^C,
\end{equation}
where the matrix field $\mathbb M_C^B = (\sg_0^{-1})^{BJ} \Big( \delta_{CJ} - \frac{\delta_{CD} v^D (\partial_J \psi_i)}{\sqrt{1-|\bar{v}|^2}} \Big)$ is invertible (in view of the gradient bound \eqref{eq:Gradient bound} for $\partial \psi_i$, the bound \eqref{eq:Smallness angular variation} for $\bar{v}$ on $U_i$ and the lower bound \eqref{eq:gS.lower.bound} for $\det(\sg_0)$); in particular, $\mathbb M^{-1}$ satisfies the pointwise bound
\begin{equation}\label{Bound M inverse}
|\mathbb M^{-1} | \le \frac{10}{1-|\rd f|_e}.
\end{equation}
Let us also recall the following basic facts about square matrices: If $A,B$ are two symmetric matrices such that $A$ is positive definite and $B$ is semi-positive definite, then
\begin{itemize}
    \item $0\le \tr(AB) \le \tr(A) \cdot \tr(B)$ (the latter bound can be computed directly by calculating the trace with respect to an orthonormal basis of eigenvectors for $A$),
    \item $AB$ has non-negative eigenvalues,
    \item $AB$ and $A^{\f12} B A^{\f12}$ have the same spectrum (but the second matrix is symmetric independently of whether $A,B$ commute or not),
    \item $\|B\|\lesssim \tr(B)$, where  $\|\cdot\|$ the Frobenius norm of a matrix and the constants implicit in $\lesssim$ depend only on the dimension of $B$.
\end{itemize}
In particular, for such matrices we can estimate
\[
\tr(B) = \tr(A^{-\f12} A^{\f12} B A^{\f12} A^{-\f12}) \le \big(\tr(A^{-\f12})\big)^2  \tr( A^{\f12} B A^{\f12}) \lesssim \|A^{-1}\|\cdot  \tr(AB)
\]
and thus
\begin{equation}\label{Trace estimate matrices}
\|AB\| \le \|A\|\cdot \|B\| \lesssim \|A\| \cdot \|A^{-1}\| \cdot \tr(AB). 
\end{equation}
Denoting with $\big[(DL^{\overline{f}(S)})^{\sharp_{\sg_0}}\big]$, $A= [\sg_0]$ and $ B= [DL^{\overline{f}(S)}]$ the $2\times 2$ matrices formed by the coordinate components of the respective tensors, we have that $A,B$ are symmetric, $A$ is positive definite, $B$ is semi-positive definite  (in view of Lemma \ref{lem:Positivity of DL}) and 
\[
\big[(DL^{\overline{f}(S)})^{\sharp_{\sg_0}}\big] = A \cdot B.
\]
Thus, applying \eqref{Trace estimate matrices}, we infer that
\begin{equation}\label{eq:Positive trace}
 \big\| \big[ (DL^{\overline{f}(S)})^{\sharp_{\sg_0}} \big] \big\| \lesssim \|[\sg_0]\|\|\sg_0^{-1}\| \cdot \tr\big[((DL^{\overline{f}(S)})^{\sharp_{\sg_0}}\big] = \|[\sg_0]\|\|\sg_0^{-1}\| \cdot \tr \chi_0 
\lesssim \f{1}{1-|\partial f|_e} \cdot \tr \chi_0.
\end{equation}
The bound \eqref{eq:Upper bound hessian} now follows by combining \eqref{Hessian psi}, \eqref{eq:Null shape again}, \eqref{Bound M inverse} and \eqref{eq:Positive trace}.
\end{proof}

To proceed, we need the following trace theorem:
\begin{lemma}\label{lem:trace}
For $U_i \subset S$ an open set among the ones defined in Lemma \ref{lem:patches}, we can estimate any smooth function $\phi: \Omega \rightarrow \mathbb R$:
\begin{equation}
    \int_{S^+ \cap U_i} |\phi|^2 \, \mathrm{dVol}_{g_0,S} \le \frac{C}{\inf_\Omega (1-|\rd f|_e)} \| \phi \|_{B^{1/2}_{2,1}(\Omega)} \label{eq:Trace}
\end{equation}
for some absolute constant $C>0$.
\end{lemma}
\begin{proof}
In view of our definition of the space $B^{\frac12}_{2,1}(\Omega)$ (see Section \ref{subsec:Besov spaces}), there exists an extension of $\phi$ on the whole of $\mathbb R^3$ such that $$ \|\phi \|_{B^{1/2}_{2,1}(\mathbb R^3)} \le 2 \| \phi\|_{B^{1/2}_{2,1}(\Omega)}.$$ We will decompose the newly extended function $\phi$ into its Littlewood--Paley pieces $\phi = \sum_{k\geq 0} \phi_k$, where $\phi_k \doteq P_k \phi$.

Let $(y^1, y^2, y^3)$ be the rotated Cartesian coordinate system associated to $U_i$, so that $U_i = \{y^3 = \psi_i(y^1, y^2) \}$. 

For each point $y_* = (y_*^1, y_*^2, y_*^3 =\psi_i(y_*^1,y_*^2)) \in S^+\cap U_i$, we can use the fundamental theorem of calculus to show that for every $y^3 \in [y_*^3-2^{-k}, y_*^3]$, we have
\begin{equation*}
\begin{split}
|\phi_k|(y_*) \leq &\: \int_{y_*^3-2^{-k}}^{y_*^3} |\rd_3 \phi_k|(y_*^1,y_*^2,y^3) \, dy^3 +  |\phi_k|(y_*^1,y_*^2,y^3). 
\end{split}
\end{equation*}
Averaging over $y^3 \in [y_*^3-2^{-k}, y_*^3]$, and then using the Cauchy--Schwarz inequality, we obtain
\begin{equation*}
\begin{split}
|\phi_k|(y_*) \leq &\: \int_{y_*^3-2^{-k}}^{y_*^3} |\rd_3 \phi_k|(y_*^1,y_*^2,y^3) \, dy^3 + 2^{k}\int_{y_*^3-2^{-k}}^{y_*^3} |\phi_k|(y_*^1,y_*^2,y^3) \, dy^3 \\
\leq &\: 2^{-k/2}\Big(\int_{y_*^3-2^{-k}}^{y_*^3} |\rd_3 \phi_k|^2(y_*^1,y_*^2,y^3) \, dy^3\Big)^{1/2} + 2^{k/2}\Big(\int_{y_*^3-2^{-k}}^{y_*^3} |\phi_k|^2(y_*^1,y_*^2,y^3) \, dx^3\Big)^{1/2}.
\end{split}
\end{equation*}
It follows that 
\begin{equation}\label{eq:trace.to be.integrated}
\begin{split}
|\phi_k|^2(y_*) \leq 2^{-k}\int_{\mathbb R} |\rd_3 \phi_k|^2(y_*^1,y_*^2,y^3) \, dy^3 + 2^{k} \int_{\mathbb R} |\phi_k|^2(y_*^1,y_*^2,y^3) \, dy^3.
\end{split}
\end{equation}

Integrating \eqref{eq:trace.to be.integrated} over $y_* \in S^+ \cap U_i$ with respect to the volume form $dy^1dy^2$, taking square roots, and then summing over $k\geq 0$, we obtain
\begin{equation}
\begin{split}
\Big(\int_{S^+\cap U_i} |\phi|^2 \,dy^1dy^2\Big)^{1/2} \leq &\: \sum_{k\geq 0} \Big(\int_{U_i} |\phi_k|^2(y_*) \, dy^1dy^2 \Big)^{1/2} \\
\leq &\: \sum_{k\geq 0} \Big(2^{-k/2} \|\rd_3 \phi_k\|_{L^2(\mathbb R^3)} + 2^{k/2} \| \phi_k\|_{L^2(\mathbb R^3)} \Big) \\
\ls &\: \sum_{k\geq 0} 2^{k/2} \| \phi_k\|_{L^2(\mathbb R^3)} = \| \phi \|_{B^{1/2}_{2,1}(\mathbb R^3)} \leq 2\| \phi \|_{B^{1/2}_{2,1}(\Omega)},
\end{split}
\end{equation}
where in the last line we have used Bernstein's inequality. The bound \eqref{eq:Trace} now follows using the fact that, on $U_i$, $\mathrm{dVol}_{g_0,S} = \det(\sg_0)dy^1dy^2$ can be controlled by \eqref{eq:gS.lower.bound}. \qedhere

\end{proof}

\begin{proposition}\label{prop:upper.bound}
For $\ep_0>0$ sufficiently small depending on $\inf_\Omg (1-|\rd f|_e)$, the following estimate holds on $S^+$:
\begin{equation}\label{eq:Final smallness}
\int_{S^+} (\tr\chi_0)^2 \, \mathrm{dVol}_{S,g_0}  \le C\ep_0^2 \frac{1+\sup_\Omega|\partial^2 f|_{e}}{\inf_\Omega(1-|\rd f|_e)^{11}}.
\end{equation}
\end{proposition}
\begin{proof}
Let $N$ and $\{U_i\}_{i=1}^N$ be as in the statement of Lemma \ref{lem:patches}. Using the pointwise bound of Lemmas~\ref{lem:Comparison geometries} and \ref{lem:Upper bound hessian}, we obtain
\begin{align}
 \int_{S^+} & (\tr\chi_0)^2 \, \mathrm{dVol}_{g_0,S} \leq \sum_{i=1}^N \int_{S^+ \cap U_i} (\tr\chi_0)^2 \, \mathrm{dVol}_{g_0,S} \nonumber \\
& \ls  \frac{1}{\inf_\Omega (1-|\rd f|_e)^{7}}\sum_{i=1}^N \int_{S^+ \cap U_i} \Big( |\partial (g - g_0)|_e^2 + |k-k_0|_e^2   \nonumber \\
 &\hphantom{\ls  \frac{1}{\inf_\Omega (1-|\rd f|_e)^2}\sum_{i=1}^N \int_{S^+\cap U_i}}  +\frac{\ep_0^2}{\inf_\Omega (1-|\rd f|_e)^3}\big( (\tr\chi_0)^2 +|\rd^2 f|_e + 1\big)\Big) \,\mathrm{dVol}_{g_0,S} .\label{eq:main.est.0}
\end{align}
Using the trace estimate of Lemma~\ref{lem:trace} and the smallness assumption for $\|g-g_0\|_{B^{3/2}_{2,1}(\Omega)}, \|k-k_0\|_{B^{1/2}_{2,1}(\Omega)}$ of Theorem~\ref{thm:main}, we have 
\begin{align}
\int_{S^+ \cap U_i} \Big( |\partial (g - g_0)|_e^2 + |k-k_0|_e^2\Big) \, \mathrm{dVol}_{g_0,S}  & \ls \frac{1}{\inf_\Omega (1-|\rd f|_e)} \big(  \|g-g_0\|_{B^{3/2}_{2,1}(\Omega)}^2 + \| k -k_0\|_{B^{1/2}_{2,1}(\Omega)}^2\big) \nonumber \\
& \ls \frac{\ep_0^2}{\inf_\Omega (1-|\rd f|_e)}.\label{eq:main.est.1}
\end{align}
Plugging \eqref{eq:main.est.1} into \eqref{eq:main.est.0} and summing over $i=1, \ldots,N$, we obtain
\begin{equation}\label{eq:H.with.H.on.RHS}
\int_{S^+} (\tr\chi_0)^2 \, \mathrm{dVol}_{g_0,S}  \ls \ep_0^2 N \frac{1+\sup_\Omega|\partial^2 f|_{e}}{\inf_\Omega (1-|\rd f|_e)^{9}} + \frac{\ep_0^2}{\inf_\Omega (1-|\rd f|_e)^{10}} \int_{S^+} (\tr\chi_0)^2 \, \mathrm{dVol}_{g_0,S}.
\end{equation}
For $\ep_0$ sufficiently small in terms of $\inf_\Omega (1-|\rd f|_e)$, we can absorb the last term on the right-hand side of \eqref{eq:H.with.H.on.RHS} to the left-hand side and obtain \eqref{eq:Final smallness} (recalling that $N\sim \frac{1}{\inf_\Omega (1-|\rd f|_e)^2}$). \qedhere
\end{proof}

Combining Proposition~\ref{prop:S+} and Proposition~\ref{prop:upper.bound}, we can now complete the proof of Theorem~\ref{thm:main}.
\begin{proof}[Proof of Theorem~\ref{thm:main}]
For $\ep_0 > 0$ sufficiently small in terms of $\| \partial^2 f\|_{L^\infty(\Omega)}$ and $\inf_\Omega (1-|\rd f|_e)$, the lower bound from Proposition~\ref{prop:S+} and the upper bound from Proposition~\ref{prop:upper.bound} are obviously incompatible,  leading to a contradiction. This finishes the proof of Theorem~\ref{thm:main}. \qedhere
\end{proof}

\section{Examples in spherical symmetry\label{sec:Examples}}

\subsection{A simple counterexample}

To prove Proposition~\ref{prop:sharp}, it suffices to use the fact that the trace theorem fails if $B^{1/2}_{2,1}(\mathbb R^3)$ is replaced by $H^{1/2}(\mathbb R^3)$. In fact in this case we can just rely on the following standard result.
\begin{lemma}\label{lem:trace.example}
There exists a sequence of smooth, spherically symmetric functions $\{\phi^{(j)}\}_{j=1}^\infty : \mathbb R^3 \rightarrow \mathbb{R}$ supported in $r\in [\f12, 2]$ such that 
$$|\phi^{(j)}|(r=1) \geq 10,\quad \|\phi^{(j)}\|_{H^{1/2}(\mathbb R^3)} \leq 2^{-j}.$$
\end{lemma}

\noindent We can now prove Proposition~\ref{prop:sharp}:
\begin{proof}[Proof of Proposition~\ref{prop:sharp}]
Let $\phi^{(j)}$ be as in Lemma~\ref{lem:trace.example}. Multiplying $\phi^{(j)}$ with $-1$ if necessary, we can assume that 
$$\phi^{(j)}(r=1) <-10.$$

We will now define the sequence $\{ (g^{(j)}, k^{(j)}) \}_{j=1}^{\infty}$ of initial data pairs on $\mathbb{R}^3$ as follows:
\begin{itemize}
\item $g^{(j)}$ is identically equal to the Euclidean metric $e$,
\item In the standard polar coordinates $(r, \vartheta, \varphi)$ on $\mathbb R^3$, the components of $k^{(j)}$ are
$$k^{(j)}_{rr} = k^{(j)}_{r\vartheta} = k^{(j)}_{\vartheta\varphi} = 0, \quad k^{(j)}_{\vartheta\vartheta} = \phi^{(j)}\quad \text{and} \quad k^{(j)}_{\varphi\varphi} = \phi^{(j)} \sin^2 \vartheta.$$
\end{itemize}

In view of Lemma~\ref{lem:trace.example}, the definition of $(g^{(j)}, k^{(j)})$ implies that
$$\|g^{(j)} -e \|_{H^{3/2}(\mathbb R^3)} + \| k^{(j)} \|_{H^{1/2}(\mathbb R^3)} \lesssim 2^{-j},$$

\noindent Now let $\Sigma = \rd B(0,1)$. Let $h$ denote the second fundamental form of $\Sigma$ in $(\mathbb R^3, g^{(j)})$. By a direct computation, $\mathrm{tr}_{g^{(j)}_\Sigma} h = 2$. On the other hand, $\mathrm{tr}_{g^{(j)}_\Sigma} k = -2 \phi^{(j)} =  -20$. Thus, $\Sigma$ is a trapped surface in $(\mathbb R^3, g^{(j)}, k^{(j)})$ for all $j \in \mathbb N$. \qedhere
\end{proof}

\subsection{A counterexample for the Einstein--scalar field system}

We will now proceed to establish Proposition \ref{prop:sharp.constraint}. 

\begin{proof}[Proof of Proposition~\ref{prop:sharp.constraint}] For any $j \in \mathbb N$, let $\psi_0^{(j)}, \psi_1^{(j)}$ be smooth, spherically symmetric functions on $B(0,1)$. Let also  $g^{(j)}$ and $k^{(j)}$ be, respectively, a spherically symmetric Riemannian metric and a spherically symmetric $(0,2)$-tensor on $B(0,1)$, expressed in polar coordinates $(\rho, \th, \varphi)$ as:
$$g^{(j)}(\rho,\varphi, \th) = \ud\rho^2 + (r^{(j)}(\rho))^2 \big( \ud\th^2 + \sin ^2 \th \ud\varphi^2 \big) $$
and
$$ k^{(j)} (\rho, \th, \varphi) = k_{\rho\rho}^{(j)}(\rho) \ud\rho^2 + k_{\th \th}^{(j)}(\rho)\big( \ud\th^2 + \sin ^2 \th \ud\varphi^2 \big),$$
where $r^{(j)}$, $k_{\rho\rho}^{(j)}$ and $ k_{\th\th}^{(j)}$ are smooth functions on $[0,1)$. In order to somewhat simplify our notations, from now on we will drop the superscript $\cdot ^{(j)}$ from $g,k$ and $\psi_0, \psi_1$ when no confusion arises.

With $(g,k; \psi_0, \psi_1)$ as above, the constraint equations \eqref{Hamiltonian constraint}--\eqref{Momentum constraint} reduce to:
\begin{align}
& -r''+\f1{2r} \Big( 1-(r')^2 \Big)  + \f1r k_{\rho\rho} k_{\th \th} + \f1{2r^3} k_{\th\th}^2  =\f14 r \big( (\psi_0')^2 + \psi_1^2 \big),   \label{Hamiltonian constraint symmetry}\\
 & k_{\rho \rho} r' - \big( \f{k_{\th\th}}{r} \big)' = \f12 r \psi_1 \psi_0'   \label{Momentum constraint symmetry}
\end{align}

We will use the following ansatz for $r$: 
$$ r^{(j)}(\rho) \doteq \rho \cdot \Big( 1 - 2^{-j-M} F(\rho)\Big)^{\f12},$$
where $M>0$ is a \emph{large} absolute constant (i.\,e. independent of $j$) and
$$ F(\rho) \doteq \int_0^\rho \chi (x) \log \big( -\log|x-1| \big) \ud x, $$
where $\chi :[0,1]\rightarrow [0,1]$ is a fixed  smooth cutoff function such that 
$$\chi \equiv 0  \quad \text{on} \quad [0, \f12], \quad\chi\equiv 1 \quad \text{on} \quad [\f34, 1] \quad \text{and} \quad \chi' \ge 0 \quad \text{on} \quad[0,1].$$
Note that, provided $M>10$, the function $r(\rho)$ belongs to $C^\infty([0,1)) \cap C^0([0,1])$ and satisfies
$$\f{r}{\rho} \ge \f12,$$
$$ r(\rho) = \rho \quad \text{for} \quad \rho \in [0,\f12]$$
and
\begin{align} \label{Positive ansatz}
-\f{r''}{r} &+ \f1{2r^2}\big( 1- (r')^2 \big)  \\
& = \f{2^{-j-M}}{8\rho^2 \big( 1-2^{-j-M} F(\rho)\big)^2} \Bigg( 4 \big( 1-2^{-j-M} F(\rho) \big)\Big( \rho^2 F''(\rho) + 3\rho F'(\rho)+F(\rho)\Big) + 2^{-j-M}(F'(\rho))^2 \Bigg) \nonumber\\
& \ge 0. \nonumber
\end{align}
We will define the coefficient $k_{\th\th}$ of $k$ by the relation
\begin{equation}
k_{\th\th}^{(j)} (\rho)  = -2^{-j-2M} \chi(2\rho) \rho, \label{k theta} 
\end{equation}
where $\chi$ was defined above. Note that the definition of the cutoff function $\chi$ implies that the support of $(k_{\th\th}/r)'$ is contained in the interval $[\f14, \f38]$.

We will define $\psi_0$ as follows:
$$ \psi_0^{(j)}(\rho) \doteq 2 \int_0^{\rho} \Big( -\f{(r^{(j)})''}{r^{(j)}} + \f1{2(r^{(j)})^2} \big( 1-((r^{(j)})')^2 \big) + \frac{(k_{\th\th}^{(j)})^2}{2(r^{(j)})^4}\Big)^{\f12}(\bar\rho)\, \ud \bar\rho.  $$
Note that $\psi_0$ is well-defined in view of \eqref{Positive ansatz}; moreover, in view of the properties of the cut-off function $\chi$ in the definition of $F$, $\psi_0$ is a smooth function on $[0,1)$ vanishing identically on $[0, \f14]$.

It remains to introduce an ansatz for the functions $\psi_1 (\rho)$ and $k_{\rho\rho}(\rho)$. The expressions for $\psi_1$ and $k_{\rho\rho}$ will be chosen to satisfy the following pair of relations:
\begin{align} \label{System h psi1}
k_{\rho\rho}^{(j)}(\rho)-2^{-j-2M-\frac{1}{2}} \chi(2\rho) \psi_1^{(j)}(\rho) & = - 2^{-j-2M+1} \chi'(2\rho),\\
\big( \psi_1^{(j)}(\rho) \big)^2  +2^{-j-2M+2} \frac{\chi(2\rho)}{\rho} k_{\rho\rho}^{(j)}(\rho) &= 0.  \nonumber 
\end{align}
In particular, we will choose
$$ \psi_1^{(j)}(\rho) \doteq -2^{-2j-4M+\f12} \frac{\chi^2(2\rho)}{\rho} + \sqrt{2^{-4j-8M+1}\frac{\chi^4(2\rho)}{\rho^2}+2^{-2j-4M+3}\frac{\chi(2\rho)}{\rho}\chi'(2\rho)}$$
and
$$ k_{\rho\rho}^{(j)}(\rho) \doteq -2^{-j-2M+1} \chi'(2\rho) + 2^{-j-2M-\f12}\chi(2\rho) \psi_1^{(j)}(\rho).$$
Note that both $\psi_1$ and $k_{\rho\rho}$ vanish outside the support of $\chi'(2\rho)$. In particular,
$$ \text{supp }\psi_1^{(j)}, \text{supp }k_{\rho\rho}^{(j)} \subseteq [\f14, \f38]. $$
 Recall that $r(\rho) = \rho$ on $[0, \f12]$. Thus, we readily deduce that $(g^{(j)}, k^{(j)}; \psi_0^{(j)}, \psi_1^{(j)})$ defined as above satisfy the constraint equations \eqref{Hamiltonian constraint symmetry}--\eqref{Momentum constraint symmetry}.
 
  Notice that, for any $j\in \mathbb N$, there exists some $\rho^{(j)}_0 \in (\f12,1) $ such that 
\begin{equation} \label{Zero normal}
(r^{(j)})' (\rho_0^{(j)}) = 0. 
\end{equation}
This can be immediately inferred from the fact that $(r^{(j)})'(\f12)=1$ and $(r^{(j)})'(1)=-\infty$. We will now show that the sphere $S^{(j)} = \{ \rho = \rho_0^{(j)} \}$ is a trapped surface for the initial data set $(g^{(j)}, k^{(j)}; \psi_0^{(j)}, \psi_1^{(j)})$: We immediately calculate that the normal $N$ of $S^{(j)}$  is $\frac{\partial}{\partial \rho}$
 and the second fundamental $h$ of $S^{(j)}$ vanishes identically, as a consequence of \eqref{Zero normal}. Therefore:
 $$ \tr_{S^{(j)}} (k-h) = \tr_{S^{(j)}} (k+h) = \tr_{S^{(j)}} k =\frac{2}{\big( r^{(j)}(\rho_0^{(j)}) \big)^{2}} k^{(j)}_{\th \th}(\rho_0^{(j)}) = - 2^{-j-2M+1}  \frac{\chi(2\rho_0^{(j)})}{\rho_0^{(j)} \big( 1 - 2^{-j-M} F(\rho_0^{(j)}) \big)}   < 0 $$
(since $\rho_0^{(j)} \in [\f12, 1]$), i.e.~$S^{(j)}$ is a trapped surface. 

We will now show that $(g,k)$ satisfy the smallness bound \eqref{Smallness counterexample}. It is straightforward to check that the components $k_{\rho\rho}$ and $k_{\th \th}$ of $k$ can be extended as $C^\infty$ functions on the closed unit ball $\{ \rho \le 1 \}$ and are supported on $\{ \rho \ge \f14 \}$, satisfying
$$ \| k^{(j)}_{\rho \rho}\|_{C^l(B(0,1))}+\| k^{(j)}_{\th \th}\|_{C^l(B(0,1))} \lesssim_l 2^{-j-2M} \quad \text{for any} \quad j,l\in \mathbb{N}.$$
As a result, $k^{(j)}$ satisfies \eqref{Smallness counterexample} provided $M$ has been chosen to be larger than some explicit constant. As for the metric $g$, we can express the components of the tensor $g-e$ in Cartesian coordinates as follows:
\begin{align*}
g^{(j)}-e & = \big( \frac{(r^{(j)}(\rho))^2}{\rho^2 } - 1\big) \rho^2 \big( \ud\th^2 + \sin^2 \th \ud\varphi^2 \big) \\
&= \big( \frac{(r^{(j)}(|x|))^2}{|x|^2 } - 1\big) \Big( \sum_{i=1}^3 (\ud x^i)^2 -\big( \sum_{i=1}^3 \f{x^i}{|x|} \ud x^i \big)^2\Big)\\
& =2^{-j-M} F(|x|) \Big( \sum_{i=1}^3 (\ud x^i)^2 -\big( \sum_{i=1}^3 \f{x^i}{|x|}\ud x^i \big)^2\Big).
\end{align*}
Thus, the proof of \eqref{Smallness counterexample} will follow once we show that the spherically symmetric function $F(|x|)$ on $B(0,1)$ has finite $H^{\f32}$ norm. Since $F(|x|)$ is a $C^\infty$ function away from $|x|=1$ and is supported on $\{ |x| \ge \f12\}$, that statement follows as a corollary of Lemma \ref{lem:Function} below. \qed   
\end{proof}

The remainder of this subsection will be devoted to the proof of Lemma~\ref{lem:Function}, which we have used above.

\begin{lemma}\label{lem:use.2D.trace}
Let $\mathcal X:\mathbb R\to [0,1]$ be a smooth cutoff function such that 
$\mathrm{supp}(\mathcal X) \subset [\f 9{10}, \f {11}{10}]$ and $\mathcal X\equiv 1$ on $[\f {19}{20}, \f {21}{20}]$. Then $\mathcal X(x- 1) \log |\log x| \in H^{\f 12}(\mathbb R)$.
\end{lemma}
\begin{proof}
Consider the function $h:\mathbb R^2 \to \mathbb R$, which is radial and given by $h(r) = \mathcal X(r - \f 1{10}) \log|\log r|$. We compute $h'(r) = \mathcal X'(r - 1) \log|\log r| + \mathcal X(r - 1) \f{1}{r \log r}$. Hence,
$$\|h \|_{H^1(\mathbb R^2)}^2 = \int_0^\infty |h'(r)|^2 r\, \ud r \leq C \int_{\f 1{20}}^{\f 1{10}}  (\log|\log r|)^2  r\, \ud r + \int_0^{\f 1{10}} \Big( \f{1}{r\log r} \Big)^2 r \, \ud r <\infty.$$
Since $h(x,y) \mapsto h(x,0)$ is a bounded map $H^1(\mathbb R^2) \to H^{\f 12}(\mathbb R)$ (by standard trace estimates), we obtain the desired result.  \qedhere
\end{proof}

By translation invariance of the $H^{\f 12}(\mathbb R)$ norm, Lemma~\ref{lem:use.2D.trace} immediately implies the following result:
\begin{lemma}\label{lem:just.translate}
Define $G:\mathbb R\to \mathbb R$ by 
\begin{equation}\label{eq:G}
G(\xi)\doteq \int_0^\infty e^{-i r \xi} \mathcal X(r) \log(-\log |r-1|) \, \ud r,
\end{equation}
where $\mathcal X$ is as in Lemma~\ref{lem:use.2D.trace}. Then 
$$\int_{-\infty}^\infty (1+ |\xi|^2)^{\f 12}|G(\xi)|^2 \, \ud \xi < \infty.$$
\end{lemma}

\begin{lemma} \label{lem:Function}
Consider the radial function $h: \mathbb R^3 \to \mathbb R^3$ given by
$$h(r) = \f 1 r \mathcal X(r) \log(-\log |r-1|),$$
where $\mathcal X$ is a cutoff function as in Lemma~\ref{lem:use.2D.trace}. Then $h \in H^{\f 12}(\mathbb R^3)$.
\end{lemma}

\begin{proof}
The Fourier transform $\mathcal F h$ is $h$ is a radial function, i.e.~can be expressed as $\mathcal F h (\xi) = \widehat{h}(|\xi|)$, where $\widehat{h}$ is expressed as the Hankel transform of $h$. Hence,
$$\widehat{h}(s) = 4\pi \int_0^\infty \f{\sin(sr)}{sr} h(r) r^2 \, \ud r = -\f{4\pi}{s} \mathfrak{Im}\Big(G(s) \Big),$$
where the function $G(s)$ was defined in  \eqref{eq:G} and $\mathfrak{Im} \Big(G(s) \Big)$ denotes its imaginary part.

We now compute
\begin{equation*}
\begin{split}
 \int_{\mathbb R^3}(1+|\xi|^2)^{\f12}\big|\mathcal F h(\xi) \big|^2 \, \ud \xi 
= 16\pi^2 \int_0^\infty (1+|s|^2)^{\f 12} \Big( \f{\mathfrak{Im}(G(s) )}{s}\Big)^2 s^2\, \ud s \leq 16\pi^2 \int_{-\infty}^\infty (1+ |s|^2)^{\f 12}|G(s)|^2 \, \ud s < + \infty
\end{split}
\end{equation*}
by Lemma~\ref{lem:just.translate}. This implies $h \in H^{\f 12}(\mathbb R^3)$. \qedhere
\end{proof}

\subsection{A spherically symmetric result in $H^{\f 32}$}

\begin{proof}[Proof of Proposition~\ref{prop:SS}]
In order for this to be a smooth metric, we must have $\f{r^2(\rho)}{\rho^2} -1 = O(\rho^2)$ as $\rho\to 0$. Taylor expanding $r$ in $\rho$, this means that 
\begin{equation}\label{eq:r.regularity}
r(\rho) = \rho + O(\rho^3).
\end{equation}

\pfstep{Step~1: Application of Sobolev embedding} Since $H^{\f 12}(B(0,R)) \hookrightarrow L^{3}(B(0,R))$, $H^{\f 32}(B(0,R)) \hookrightarrow W^{1,3}(B(0,R))$ (with constants independent of $R$), the assumptions on $k$, $\psi_0$ and $\psi_1$ imply that
\begin{equation}\label{eq:SS.embedding}
    \int_0^R \Big(|k_{\rho\rho}|^3 + |\f{k_{\th\th}}{\rho^2}|^3 + |\psi_0'|^3 + |\psi_1|^3\Big)\, \rho^2 \ud \rho \ls \ep_0^3.
\end{equation}

\pfstep{Step~2: Controlling the geometry} Define the \emph{Hawking mass} by 
\begin{equation}\label{eq:Hawking}
    m = \f r2 + \f 12 \f{k_{\th\th}^2}r - \f{r (r')^2}2.
\end{equation}

For $\rho_* \in (0,R]$, assume that the following \textbf{bootstrap assumptions} hold for all $\rho \in [0, \rho_*)$:
\begin{itemize}
    \item 
    \begin{equation}\label{eq:BA1}
        \f{\rho}2 \leq r(\rho) \leq 2 \rho.
    \end{equation}
    \item 
    \begin{equation}\label{eq:BA2}
        \f {2m}r \leq \f 12,
    \end{equation}
\end{itemize}
By smoothness at $\rho =0$, we know that \eqref{eq:BA1} and \eqref{eq:BA2} hold for small $\rho$. Our \textbf{goal} will be to show that 
\begin{itemize}
    \item 
    \begin{equation}\label{eq:BA1.goal}
        \f{\rho}{\sqrt{2}} \leq r(\rho) \leq \sqrt{2} \rho.
    \end{equation}
    \item 
    \begin{equation}\label{eq:BA2.goal}
        \f {2m}r \leq \f 14,
    \end{equation}
\end{itemize}
Once we prove the bounds in \eqref{eq:BA1.goal} and \eqref{eq:BA2.goal} under the bootstrap assumptions \eqref{eq:BA1} and \eqref{eq:BA2}, a standard continuity argument shows that in fact both \eqref{eq:BA1.goal} and \eqref{eq:BA2.goal} hold for all $\rho \in [0, R]$.

\pfstep{Step~2(a): Non-negativity of $m$} Using the definition of $m$ in \eqref{eq:Hawking} together with the constraint equations \eqref{Hamiltonian constraint symmetry}--\eqref{Momentum constraint symmetry}, we obtain
\begin{equation}\label{eq:mass.equation}
    \begin{split}
        m' = &\: \Big(\f r2 + \f 12 \f{k_{\th\th}^2}r - \f{r (r')^2}2 \Big)' \\
        = &\: \f{r'}2 - \f 12 r' \f{k_{\th\th}^2}{r^2} + r \f{k_{\th\th}}{r} (k_{\rho\rho} r' - \f 12 r \psi_1 \psi_0') \\
        &\: - \f{(r')^3}2 - rr' \Big(\f 1{2r} (1 - (r')^2) + \f 1r k_{\rho\rho} k_{\th\th} + \f 1{2r^3} k_{\th\th}^2 - \f 14 r ((\psi_0')^2 + \psi_1^2) \Big) \\
        = &\: \f 12 r' \f{k_{\th\th}^2}{r^2} + r \f{k_{\th\th}}r (k_{\rho\rho} r' - \f 12 r \psi_1 \psi_0') -rr'\Big( \f 1r k_{\rho\rho} k_{\th\th} + \f 1{2r^3} k_{\th\th}^2 - \f 14 r ((\psi_0')^2 + \psi_1^2)\Big)\\
        = &\:  -r k_{\th\th}  \psi_1 \psi_0' + r^2 r'\f 14 \Big((\psi_0')^2 + \psi_1^2\Big)\\
        = &\: \f{r^2}8\Big(r'+ \f{k_{\th\th}}{r}\Big) \Big(\psi_1 - \psi_0' \Big)^2 + \f{r^2}8 \Big(r'- \f{k_{\th\th}}{r}\Big)\Big(\psi_1 + \psi_0' \Big)^2.
    \end{split}
\end{equation}

By \eqref{eq:Hawking} and \eqref{eq:BA2}, we have 
\begin{equation}\label{eq:r'.lower}
    (r')^2 = (1-\f{2m}r) + \f{k_{\th\th}^2}{r^2} \geq \f 12.
\end{equation}
Using \eqref{eq:r.regularity}, we have $r'(0) =1$. Hence, continuity of $r'$ and \eqref{eq:r'.lower} imply that $r' \geq \f{1}{\sqrt{2}}$.

By \eqref{eq:BA2}, $1 + \f 12 \f{k_{\th\th}^2}{r^2}- \f{(r')^2}2\leq \f 12$. In particular, using also the positivity of $r'$ that we just established,
\begin{equation}\label{eq:kthth/rleqr'}
    |\f{k_{\th\th}}{r}|< |r'| = r'.
\end{equation}

Hence, every term on the right-hand side of \eqref{eq:mass.equation} is non-negative. Since regularity implies that $m=0$ at $\rho = 0$, we obtain $m \geq 0$ for all $\rho \in [0,\rho_*)$.

\pfstep{Step~2(b): Proof of \eqref{eq:BA1.goal}} The lower bound in \eqref{eq:BA1.goal} is easier. Indeed, \eqref{eq:r'.lower} implies (using \eqref{eq:r.regularity}) that
\begin{equation}\label{eq:BA1.goal.lower}
    r(\rho) = \int_0^\rho r'(\rho)\, \ud \rho \geq \f 1{\sqrt{2}}\int_0^\rho  \ud \rho = \f \rho{\sqrt{2}}.
\end{equation}

To obtain the upper bound in \eqref{eq:BA1.goal}, we first need an improved estimate for $k_{\th\th}$. For this, note that by \eqref{eq:BA1} and \eqref{eq:SS.embedding}, we have $\int_0^R |\f{k_{\th\th}}{r}|^3\, \rho^{-1} \ud \rho \ls \ep_0$. Hence, by the Cauchy--Schwarz inequality,
\begin{equation}\label{eq:kthth.gain.log}
\int_0^\rho |\f{k_{\th\th}}{r}| \, \ud \bar\rho = \int_0^\rho |\f{k_{\th\th}}{r}| \f{\rho^{1/3}}{\rho^{1/3}}\, \ud \bar\rho \leq (\int_0^\rho |\f{k_{\th\th}}{r}|^3 \rho^{-1} \, \ud \bar\rho)^{\f 13} (\int_0^\rho \bar\rho^{1/2}\, \ud \bar\rho)^{2/3} \ls \ep_0 \rho.
\end{equation}
Now, we use \eqref{eq:Hawking} together with $m\geq 0$ (established in Step~2(a)) and \eqref{eq:kthth.gain.log} to obtain
\begin{equation}\label{eq:r'.upper}
    (r')^2 = 1- \f{2m}r + \f{k_{\th\th}^2}{r^2} \leq 1 + \f{k_{\th\th}^2}{r^2},
\end{equation}
which implies 
\begin{equation}\label{eq:BA1.goal.upper}
    r(\rho) = \int_0^\rho r'(\rho)\, \ud \rho \leq \int_0^\rho \sqrt{1 + \f{k_{\th\th}^2}{r^2}} \, \ud \rho \leq \int_0^\rho \Big( 1 + |\f{k_{\th\th}}{r}| \Big) \, \ud \rho \leq \rho + O(\ep_0)\rho.
\end{equation}

The estimates \eqref{eq:BA1.goal.lower} and \eqref{eq:BA1.goal.upper} imply the lower and upper bound in \eqref{eq:BA1.goal} respectively.

\pfstep{Step~2(c): Improved bound for $r'$} In this step, we derive a bound for $r'(\rho)$ (see already \eqref{eq:r'.improved}), which will be important in Step~2(d) below.

\pfstep{Step~2(c).i: An estimate for $\f{k_{\th\th}}{r}$} By the constraint equation \eqref{Momentum constraint symmetry}, 
\begin{equation}\label{eq:kthth/r.constraint}
    \Big|(\f{k_{\th\th}}{r})' \Big| \leq |k_{\rho\rho} r'| + |r\psi_1\psi_0'|. 
\end{equation}

We control each term on the right-hand side of \eqref{eq:kthth/r.constraint}. By the Cauchy--Schwarz inequality and \eqref{eq:SS.embedding}, we have
\begin{equation}\label{eq:kthth/r.constraint.1}
    \int_0^\rho |k_{\rho\rho} r'|^{3/2} \, \bar\rho^2 \ud \bar\rho \ls \Big(\int_0^\rho |k_{\rho\rho}|^{3}\,\bar{\rho}^2 \ud \bar\rho \Big)^{1/2} \Big(\int_0^\rho |r'|^{3}\, \bar{\rho}^2 \ud \bar\rho\Big)^{1/2} \ls \ep_0^{\f 32} \Big(\int_0^\rho |r'|^{3}\, \bar{\rho}^2 \ud \bar\rho\Big)^{1/2}.
\end{equation}
By \eqref{eq:BA1}, the Cauchy--Schwarz inequality and \eqref{eq:SS.embedding}, we have
\begin{equation}\label{eq:kthth/r.constraint.2}
    \int_0^\rho |r\psi_1\psi_0'|^{3/2} \,  \bar\rho^2 \ud \bar\rho \leq 2^{3/2} \rho^{3/2} \Big( \int_0^\rho |\psi_1|^3 \, \bar\rho^2 \, \ud \bar\rho \Big)^{1/2} \Big(\int_0^\rho |\psi_0'|^3 \, \bar\rho^2 \ud \bar\rho \Big)^{1/2} \ls \ep_0^{3} \rho^{3/2}.
\end{equation}
Hence, plugging the bounds \eqref{eq:kthth/r.constraint.1} and \eqref{eq:kthth/r.constraint.2} into \eqref{eq:kthth/r.constraint}, we obtain
\begin{equation}
    \int_0^\rho |(\f{k_{\th\th}}{r})'|^{3/2} \,  \bar\rho^2 \ud \bar\rho \ls \ep_0^3 \rho^{3/2} + \ep_0^{\f 32}\Big(\int_0^\rho |r'|^{3}\, \bar{\rho}^2 \ud \bar\rho\Big)^{1/2}.
\end{equation}
Using the $\dot{W}^{1,\f 32}(B(0,R)) \hookrightarrow L^3(B(0,R))$ Sobolev embedding, we then obtain
\begin{equation}\label{eq:kthth/r}
    \int_0^\rho |\f{k_{\th\th}}{r}|^{3} \,  \bar\rho^2 \ud \bar\rho \ls \Big(\int_0^\rho |(\f{k_{\th\th}}{r})'|^{3/2} \,  \bar\rho^2 \ud \bar\rho\Big)^2\ls \ep_0^6 \rho^{3} + \ep_0^{3}\Big(\int_0^\rho |r'|^{3}\, \bar{\rho}^2 \ud \bar\rho\Big).
\end{equation}

\pfstep{Step~2(c).ii: Proof of the improved estimate for $r'$} By \eqref{eq:r'.upper} and \eqref{eq:kthth/r}, we have
\begin{equation}
\begin{split}
    \int_0^\rho |r'|^{3} \,  \bar\rho^2 \ud \bar\rho \ls \int_0^\rho \Big(1 + \big|\f{k_{\th\th}}{r} \big|^3 \Big) \,  \bar\rho^2 \ud \bar\rho \ls \rho^3 + \ep_0^{3}\Big(\int_0^\rho |r'|^{3}\, \bar{\rho}^2 \ud \bar\rho\Big).   
\end{split}
\end{equation}
For $\ep_0$ sufficiently small, we can absorb the last term on the right-hand side to the left to obtain
\begin{equation}\label{eq:r'.improved}
\begin{split}
    \int_0^\rho |r'|^{3} \,  \bar\rho^2 \ud \bar\rho \ls \rho^3.   
\end{split}
\end{equation}

\pfstep{Step~2(d): Proof of \eqref{eq:BA2.goal}} Plugging the bounds \eqref{eq:BA1} and \eqref{eq:kthth/rleqr'} into \eqref{eq:mass.equation}, we obtain
\begin{equation}\label{eq:m'.estimated}
    |m'(\rho)| \leq 2\rho^2 |r'| \Big( \psi_1^2 + (\psi_0')^2 \Big).
\end{equation}
Integrating \eqref{eq:m'.estimated} starting from the regularity condition $m(\rho=0) = 0$, we obtain
\begin{equation}
    m(\rho) \leq 2\int_0^\rho \bar\rho^2 |r'| \Big( \psi_1^2 + (\psi_0')^2 \Big) \, \ud \bar\rho.
\end{equation}
for all $\rho \in [0, \rho_*)$.

Using H\"older's inequality and then \eqref{eq:SS.embedding} and \eqref{eq:r'.improved}, we obtain
\begin{equation}\label{eq:BA2.goal.almost}
    m(\rho) \ls \Bigg(\int_0^\rho \Big(|\psi_0'|^3 + |\psi_1|^3 \Big)\, \bar\rho^2 \ud \bar\rho \Bigg)^{2/3} \Big(\int_0^\rho |r'|^{3} \,  \bar\rho^2 \ud \bar\rho \Big)^{1/3} \ls \ep_0^2 \rho. 
\end{equation}

Finally, combining \eqref{eq:BA2.goal.almost} with \eqref{eq:BA1}, we obtain
$$\f mr(\rho) \ls \ep_0^2$$
for all $\rho \in [0,\rho_*)$. After choosing $\ep_0$ to be smaller, we have thus proven \eqref{eq:BA2.goal} and concluded the bootstrap argument.

\pfstep{Step~3: Conclusion of the argument} Having controlled the geometry in Step~2, the conclusion now follows straightforwardly. 

Indeed, denote by $S_\rho$ the $2$-sphere of constant $\rho \in (0,R)$. We first compute 
\begin{equation}
\mathrm{tr}_{S_\rho}(k \pm h) = \f{2}{r^2(\rho)} k_{\th\th} \pm \f{2 r'(\rho)}{r(\rho)}.    
\end{equation}
Now, in the course of the bootstrap argument in Step~2, we have obtained the bound \eqref{eq:kthth/rleqr'}:
$$|\f{k_{\th\th}}{r}|< |r'|.$$
Moreover, since $\lim_{\rho\to 0^+} r'(\rho) = 1$ (by \eqref{eq:r.regularity}) and $|r'(\rho)|\geq \f{1}{\sqrt{2}}$ (by \eqref{eq:r'.lower}), we have $r'(\rho) \geq \f 1{\sqrt{2}} >0$ for all $\rho \in [0,R)$. As a result, we know that $\mathrm{tr}_{S_\rho}(k + h) >0$ and $\mathrm{tr}_{S_\rho}(k - h) <0$, i.e., $S_\rho$ is not trapped. Since $\rho \in (0,R)$ is arbitrary, we have completed the proof. \qedhere
\end{proof}

\begin{remark*}
Notice that we did not use the full strength of the smallness assumption \eqref{eq:SS.smallness}. In fact, in this setting, it suffices to assume the weaker smallness assumption \eqref{eq:SS.embedding}.
\end{remark*}

\bibliographystyle{DLplain}
\bibliography{Notrappedsurfaces}

\end{document}